\documentclass[11pt, article]{IEEEtran}
\onecolumn
\usepackage{amsmath,amsfonts,amssymb,amscd,amsthm,dsfont}
\usepackage{graphicx,epsfig}
\usepackage{color}
\usepackage{mathtools}
\usepackage{bbm}
\usepackage{url}
\usepackage{balance}
\usepackage{epstopdf}
\usepackage{subfigure,enumitem}
\usepackage{cite,url}
\usepackage{verbatim}
\usepackage{titlesec}
\usepackage{bm}

\newtheorem{definition}{\textbf{Definition}}

\newtheorem{lemma}{\textbf{Lemma}}
\newtheorem{theorem}{\textbf{Theorem}}
\newtheorem{proposition}{\textbf{Proposition}}

\newcommand{\nn}{\nonumber}

\newcommand\redsout{\bgroup\markoverwith{\textcolor{red}{\rule[0.5ex]{2pt}{0.8pt}}}\ULon}
\newcommand{\zsfd}[1]{\ifmmode\text{\redsout{\ensuremath{#1}}}\else\redsout{#1}\fi}

\usepackage{color}

\begin{document}
	
\title{Kernel Robust Hypothesis Testing}
\author{Zhongchang Sun \quad Shaofeng Zou
	\thanks{ This paper was presented in part at the 2021 IEEE International Symposium on Information Theory \cite{sun2021datadriven} and the 2022 IEEE International Symposium on Information Theory \cite{sun2022robust}.}
	\thanks{Zhongchang Sun and Shaofeng Zou are with the Department of Electrical Engineering, University at Buffalo, Buffalo, NY 14228 USA (e-mail: zhongcha@buffalo.edu, szou3@buffalo.edu).}
}
\maketitle
\begin{abstract}
The problem of robust hypothesis testing is studied, where under the null and the alternative hypotheses,  the data-generating distributions are assumed to be in some uncertainty sets, and the goal is to design a test that performs well under the worst-case distributions over the uncertainty sets. In this paper, uncertainty sets are constructed in a data-driven manner using kernel method, i.e., they are centered around empirical distributions of training samples from the null and alternative hypotheses, respectively; and are constrained via the distance between kernel mean embeddings of distributions in the reproducing kernel Hilbert space, i.e., maximum mean discrepancy (MMD). 
The Bayesian setting and the Neyman-Pearson setting are investigated. For the Bayesian setting where the goal is to minimize the worst-case error probability, an optimal test is firstly obtained when the alphabet is finite. When the alphabet is infinite, a tractable approximation is proposed to quantify the worst-case average error probability, and a kernel smoothing method is further applied to design test that generalizes to unseen samples. A direct robust kernel test is also proposed and proved to be exponentially consistent. For the Neyman-Pearson setting, where the goal is to minimize the worst-case probability of miss detection subject to a constraint on the worst-case probability of false alarm, an efficient robust kernel test is proposed and is shown to be asymptotically optimal. Numerical results are provided to demonstrate the performance of the proposed robust tests.
\end{abstract}

\section{Introduction}
Hypothesis testing is a fundamental problem in statistical inference where the goal is to distinguish among different hypotheses with a small probability of error \cite{moulin2018statistical,cover2006information,kay1993fundamentals}. The likelihood ratio test is known to be optimal under different settings, e.g., the Neyman-Pearson setting and the Bayesian setting \cite{moulin2018statistical, kay1993fundamentals}. For example, for binary hypothesis testing, we compare the likelihood ratio between the two hypotheses with a pre-specified threshold to make the decision. Therefore, the data-generating distributions under different hypotheses are needed. In practice, these distributions are usually estimated from historical data or designed using domain knowledge, and thus may deviate from the true data-generating distributions. When the distributions applied in the likelihood ratio test deviate from the true data-generating distributions, the performance of the test may degrade significantly. To address this problem, the approach of robust hypothesis testing is proposed, e.g.,  \cite{huber1965robust, levy2009robust,gul2017minimax,gao2018wasserstein, xie2021robust, wang2022datadriven, barni2013source, jin2020adversarial, rieder1977least, osterreicher1978on, bednarski1981on, hafner1982prokhorov, kassam1981robust,hafner1993construction,vastola1984ppoint, faub2021review}, where uncertainty sets are introduced to model the uncertainty in the underlying distributions. Generally, the uncertainty sets are constructed as collections of distributions that lie in the neighborhood of nominal distributions based on some distance measure. The goal is to design a test that performs well under the worst-case distributions over the uncertainty sets.

The robust hypothesis testing problem has been widely studied and various ways of constructing uncertainty sets have been introduced (see, e.g.,  \cite{faub2021review, gul2017minimax} for a review). The $\epsilon$-contamination uncertainty sets and the total variation uncertainty sets were investigated in \cite{huber1965robust} and a censored likelihood ratio test was constructed and shown to be minimax optimal. The problem with uncertainty sets defined via the Kullback-Leibler (KL) divergence was investigated in \cite{levy2009robust, gul2017minimax}. The least-favorable distributions (LFDs) were identified under some conditions, and the robust likelihood ratio test based on the LFDs were constructed. In \cite{barni2013source}, the robust hypothesis testing problem under the Bernoulli distribution was investigated. In \cite{jin2020adversarial}, the uncertainty sets were constructed via  distortion constraints. In those works, nominal distributions are usually 
estimated from historical data. However, when it comes to the high-dimensional data, which is common in the big data era, it is in general difficult to obtain an accurate estimate of the data-generating distributions. Existing studies are mostly limited to the 1-dimensional case, and a generalization to high-dimensional data, e.g., finding the LFDs, is still an open problem in the literature. 

In this paper, we employ a data-driven approach \cite{gao2018wasserstein, xie2021robust} to construct the nominal distributions, and extend the robust hypothesis testing problem to the high-dimensional setting. Specifically, a number of training samples are available from the null and alternative hypotheses, respectively, and their empirical distributions are used as the nominal distributions to design the uncertainty sets. We note that in this case, the uncertainty sets defined via the KL divergence \cite{levy2009robust, gul2017minimax} are not applicable, since such uncertainty sets only contain distributions supported on the training samples, which may be problematic if the alphabet is actually infinite. 
	
	
In \cite{gao2018wasserstein, xie2021robust}, the robust hypothesis testing problem was investigated where uncertainty sets are centered around empirical distributions via the Wasserstein distance. In \cite{gao2018wasserstein}, the original 0-1 loss, i.e., the error probability, was firstly smoothed. Then, this relaxed formulation can be solved efficiently, and the LFDs and the nearly-optimal robust detector were identified. In \cite{xie2021robust}, the minimax problem with the 0-1 loss, i.e., the exact probability of error, was considered, where a computationally tractable reformulation and the optimal robust test were characterized. In \cite{wang2022datadriven}, the data-driven robust hypothesis testing problem with the Sinkhorn distance, which is a variant of Wasserstein distance with entropic regularization, was studied. The original 0-1 loss, i.e., the error probability, was smoothed as in \cite{gao2018wasserstein}. Then, a finite-dimensional convex optimization problem was proposed to approximate the smoothed problem. The solutions were further used to approximate the LFDs, and design the robust test. However, Wasserstein distance based approach has certain drawbacks. First, the Wasserstein distance between the empirical distribution with $m$ samples and its data-generating distribution is bounded by $\mathcal O(m^{-1/d})$ \cite{fournier2015rate}, which depends on the dimension $d$ of the data. Therefore, when choosing radii of uncertainty sets to guarantee that the true data-generating distributions lie in the uncertainty sets with high probability, it is too pessimistic when $d$ is large. Moreover, coefficients in such a concentration bound depend on the true distribution which is unknown, and thus makes it difficult to use in practice. Second, Wasserstein distance is computationally expensive, especially in the high-dimensional setting.

Moment information, such as mean and variance, is usually used to measure the difference between distributions. In \cite{pandit2004moment}, the uncertainty sets are constructed using moment classes, where a finite alphabet was considered, and an asymptotically optimal test was designed. Specifically, the moment uncertainty sets in \cite{pandit2004moment} are defined as $\{P: E_P[f]\leq \theta\}$, where $f$ is a real-valued function, $E_P[f]$ denotes the expectation of $f$ under $P$, and $\theta$ is a constant. In this paper, we generalize the moment classes to the reproducing kernel Hilbert space (RKHS) \cite{Berl2004,Srip2010,gretton2012kernel} and construct uncertainty sets using the maximum mean discrepancy (MMD). Specifically, let $f = g - E_{\hat{P}}[g]$, where $\hat{P}$ is the empirical distribution of samples from $P$, $g$ is any function in the RKHS. We then consider the worst-case and take the supremum of $g$ with bounded norm over the RKHS. This leads to uncertainty sets centered at $\hat{P}$ and defined by MMD (see more details in Section \ref{sec:model}). Compared with the Wasserstein distance, the kernel MMD between the empirical distribution with $m$ samples and its data-generating distribution can be bounded by  $\mathcal O(1/\sqrt{m})$ \cite{altun2006unifying, szabo2015twostage}, which is dimension-free and also this bound does not depend on the data-generating distribution. This makes it much easier to choose the radii of the uncertainty sets. Moreover, the kernel MMD is computationally efficient to evaluate.

The MMD-based test statistic has been widely used in statistical signal processing and machine learning. For the one-sample testing problem, where the goal is to distinguish if a sequence of samples come from a certain distribution, and the two-sample testing problem, where the goal is to distinguish if two sequences of samples come from the same distribution, the MMD-based methods and some variants are proposed in \cite{chwi2015fast, fuku2009kernel, gretton2009fast, gretton2012optimal, suth2017generative, zaremba2013btest,zhu2021asymptotically, lloyd2015statistical, kim2016examples}. In \cite{zou2017stream} and \cite{zou2017structures}, MMD-based tests are proposed to detect anomalous data streams and anomalous network structures, respectively, where the anomalous samples generated from a different distribution from the normal samples. In \cite{li2015mstatistic}, an MMD-based M-statistic is proposed for data-driven quickest change detection. In our paper, we apply MMD to design robust tests which perform well under the worst-case distributions over the uncertainty sets.

\subsection{Main Contributions}
In this paper, we develop a data-driven approach with the kernel method to design uncertainty sets for the problem of robust hypothesis testing. Specifically, empirical distributions are used directly as the nominal distributions, which avoids the estimation error when fitting the data into a parametric family of distributions. We then use the kernel MMD as the distance metric, which can be viewed as a generalization of the moment classes \cite{pandit2004moment}. The advantage of the kernel method is that it scales well for high-dimensional data, and choosing radii of the uncertainty set does not require the knowledge of the underlying true distribution. More importantly, our designed uncertainty sets contain continuous distributions (not only distributions supported on the training data), and our robust kernel test generalizes with guaranteed out-of-sample performance. 

We first focus on the Bayesian setting where the goal is to minimize the worst-case error probability. We first study the case with a finite alphabet, and reformulate the original problem equivalently to a finite-dimensional convex optimization problem via the strong duality of kernel robust optimization \cite{zhu2021kernel} and then derive the optimal robust test. For the case with an infinite alphabet, we propose a tractable approximation to quantify the worst-case error probability. The basic idea is to generate a finite number of samples randomly, and reduce the uncertainty set to be supported on these samples. We then rewrite equivalently the original problem as a convex optimization problem, and  approximate the original objective function value using the approximated uncertainty set supported on these randomly generated samples. 
This approximation is tractable since it is a finite-dimensional convex optimization, and it builds connection between the finite-alphabet case and the infinite-alphabet case. We then show that the solutions to the approximation converge almost surely to the solutions to the original infinite-alphabet problem as the number of randomly generated samples goes to infinity. The LFDs (for the approximation) can be recovered, which are also supported on these randomly generated samples. 
To generalize to unseen data, we further apply the kernel smoothing method on the LFDs, and design a robust test that is the likelihood ratio test between the smoothed LFDs. The computational complexity lies in solving a finite-dimensional convex optimization problem the complexity of which depends on the number of randomly generated samples, and implementing the test using kernel smoothed LFDs the complexity of which is quadratic in the number of randomly generated samples and testing samples. We also propose a direct robust kernel test that can be implemented with a quadratic complexity in the number of samples, and show that it is exponentially consistent. The basic idea is to compare the closest MMD distance between the empirical distribution of samples and the two uncertainty sets. 

We then study the Neyman-Pearson setting, where the goal is to minimize the worst-case probability of miss detection subject to a constraint on the worst-case probability of false alarm. We first develop the universal upper bound on the error exponent of miss detection under the Neyman-Pearson setting. The analysis is based on a generalization of the Chernoff-Stein lemma \cite{cover2006information, dembo2009large}. We then design a novel robust kernel test, which is to compare the closest distance between the empirical distribution of the test samples and the uncertainty set with a threshold. We further demonstrate that it is asymptotically optimal under the Neyman-Pearson setting. Our proposed robust kernel test does not need to solve for the LFDs, which might be computationally intractable in practice. We also show that our test can be implemented efficiently with a quadratic complexity in the number of samples.

\subsection{Paper Organization}
In Section \ref{sec:model}, we present the preliminaries on MMD and the problem formulation. In Section \ref{sec:baysetting}, we focus the Bayesian setting, and derive the optimal test for the case with a finite alphabet. For the case with an infinite alphabet, we provide a tractable approximation to quantify the worst-case error probability and propose a kernel smoothing robust test. We also propose an exponentially consistent direct robust kernel test. In Section \ref{sec:npsetting}, we study the robust hypothesis testing under the Neyman-Pearson setting, and propose an asymptotically optimal robust kernel test. In Section \ref{sec:simulation}, we provide numerical results to validate our theoretical analysis. In Section \ref{sec:conclusion}, we present some concluding remarks.
	
\section{Preliminaries and Problem Formulation}\label{sec:model}
Let $\mathcal{X}\subset \mathbb R^d$ be a compact set where samples are taken from. Denote by $\mathcal{P}$ the set of all probability measures on $\mathcal{X}$.

\subsection{Maximum Mean Discrepancy (MMD)}

We first give a brief introduction to the kernel mean embedding and the MMD \cite{Berl2004,Srip2010}. 
Let $\mathcal H$ denote the RKHS associated with a kernel $k(\cdot,\cdot): \mathcal{X}\times\mathcal{X}\rightarrow\mathbb R$. Specifically,  $k(x,\cdot)$ denotes the feature map: $\mathcal{X}\rightarrow\mathcal{H}$, and $k(x, y) = \langle k(x,\cdot), k(y, \cdot)\rangle_{\mathcal{H}}$ defines an inner product on $\mathcal{H}$. In this paper, we consider the bounded kernel: $0\leq k(x,x')\leq K$, $\forall x,x'\in\mathcal X$, where $K>0$ is some positive constant. The kernel mean embedding of a distribution is a mapping from $\mathcal P$ to $\mathcal H$ defined as $\mu_P = \int k(x,\cdot) dP$. 
Let $E_P[f]$ denote the expectation of  a function $f\in\mathcal H$.  Denote by $\|\cdot\|_{\mathcal{H}}$ the norm on $\mathcal{H}$. 
Define the MMD between two distributions $P_0$ and $P_1$ as:
\begin{flalign}
d_{\text{MMD}}(P_0,P_1) =\sup_{f\in\mathcal{H}: \|f\|_{\mathcal{H}}\leq 1}E_{P_0}[f(x)]-E_{P_1}[f(x)].
\end{flalign}
With the reproducing property of the RKHS, we have that $E_P[f] = \langle f, \mu_P\rangle_{\mathcal{H}}$. The MMD between $P_0$ and $P_1$ can be equivalently written as the distance between $\mu_{P_0}$ and $\mu_{P_1}$ in the RKHS \cite{gretton2012kernel}:
\begin{flalign}\label{eq:mmddef}
&d_{\text{MMD}}(P_0,P_1) = \big\|\mu_{P_0} - \mu_{P_1}\big\|_{\mathcal{H}}\nn\\& = \Big(E_{x\sim P_0, x^\prime\sim P_0}[k(x,x^\prime)] + E_{y\sim P_1, y^\prime\sim P_1}[k(y,y^\prime)] - 2E_{x\sim P_0, y\sim P_1}[k(x,y)]\Big)^{1/2}.
\end{flalign}
Given samples $x^n = (x_1, x_2, \cdots, x_n) \sim P_0$ and $y^m = (y_1, y_2, \cdots, y_m) \sim P_1$, an unbiased estimate of the squared MMD \cite{gretton2012kernel} between $P_0$ and $P_1$ is 
\begin{flalign}
&\hat {d^2}_{\text{MMD}}(P_0,P_1) = \frac{1}{n(n-1)}\sum_{i=1}^n\sum_{j\neq i} k(x_i, x_j) + \frac{1}{m(m-1)}\sum_{i=1}^m\sum_{j\neq i} k(y_i, y_j) -\frac{2}{nm}\sum_{i=1}^n\sum_{j=1}^mk(x_i, y_j).
\end{flalign}
If a kernel $k$ is characteristic \cite{muandet2017kernel}, the kernel mean embedding  is injective, and then $d_{\text{MMD}}(\cdot,\cdot)$ is a metric on $\mathcal{P}$ \cite{gretton2012kernel, sripe2010hilbert}. In this paper, we consider kernels such that the weak convergence on $\mathcal{P}$ is metrized by MMD \cite{simon2018kernel, sripe2016weak}, e.g., Gaussian kernels and Laplacian kernels.
	
\subsection{Problem Setup}\label{sec:introsetup}
Let $\mathcal{P}_0, \mathcal{P}_1\subseteq \mathcal P $ denote the uncertainty sets under the null and alternative hypotheses, respectively.
We propose a data-driven approach to construct the uncertainty sets. Instead of fitting nominal probability distributions in a parametric form, we have two sequence of training samples: $\hat{x}_0^m = (\hat{x}_{0,1}, \hat{x}_{0,2}, \cdots, \hat{x}_{0,m})$ and $\hat{x}^m_1 = (\hat{x}_{1,1}, \hat{x}_{1,2}, \cdots, \hat{x}_{1,m})$ from the two hypotheses, respectively. Let $\hat{Q}^l_m = \frac{1}{m}\sum_{i=1}^m\delta_{\hat{x}_{l,i}},$ be the empirical distribution of $\hat{x}^m_l$, $l = 0,1$, where $\delta_{\hat{x}_{l,i}}$ denotes the Dirac measure on $\hat{x}_{l,i}$. 
The nominal distributions are then the empirical distributions of data from the two hypotheses, respectively.
The uncertainty sets $\mathcal{P}_0, \mathcal{P}_1$ are defined via the MMD:
\begin{flalign}
\mathcal{P}_l = \Big\{P\in\mathcal{P}: d_{\text{MMD}}(P, \hat{Q}^l_m) \leq \theta\Big\},\ l = 0, 1,
\end{flalign}
where $\theta$ is the pre-specified radius of the uncertainty sets, and shall be chosen to guarantee that the population distribution falls into the uncertainty sets with high probability. It is assumed that $\mathcal{P}_0, \mathcal{P}_1$ do not overlap, i.e., $\theta < \frac{\|\mu_{\hat{Q}^1_m} - \mu_{\hat{Q}^0_m}\|_{\mathcal{H}}}{2}$. Otherwise, the problem is trivial.

In \cite{pandit2004moment}, the moment class is defined as $\{P\in\mathcal{P}:  E_P[f] \leq \theta\}$, where $f$ is real-value function on $\mathcal{X}$. In the definition of moment class, if we let $f = g - E_{\hat{Q}_m^l}[g]$ and take the supremum over $g$ with $\|g\|_{\mathcal H}\leq 1$ in the RKHS, it is then the MMD between $P$ and $\hat{Q}_m^l$. Therefore, the MMD uncertainty sets can be viewed as a generalization of moment classes to the RKHS. 

In this paper, we focus on the robust hypothesis testing problem with MMD uncertainty sets under the Bayesian setting and the Neyman-Pearson setting. 

\textit{1) Bayesian Setting}. Given a sample $x$ following an unknown distribution $Q$, the goal is to distinguish between the null hypothesis $H_0:  Q\in \mathcal{P}_0$ and the alternative hypothesis $H_1: Q\in \mathcal{P}_1$. For a randomized test $\phi:\mathcal{X}\rightarrow [0,1]$, it accepts the null hypothesis $H_0$ with probability $1-\phi(x)$ and accepts the alternative hypothesis $H_1$ with probability $\phi(x)$. Let 
\begin{flalign}
&P_F(\phi) \triangleq \sup_{P_0\in\mathcal{P}_0}E_{P_0}[\phi(x)],\nn\\& P_M(\phi) \triangleq\sup_{P_1\in\mathcal{P}_1}E_{P_1}[1-\phi(x)]
\end{flalign}
denote the worst-case probability of false alarm (type-\uppercase\expandafter{\romannumeral1} error probability) and the worst-case probability of miss detection (type-\uppercase\expandafter{\romannumeral2} error probability) for the test $\phi$. 

For the simple hypothesis testing with equal priors on the two hypotheses, the error probability in the Bayesian setting is given by
\begin{flalign}
P_E(\phi) &\triangleq \frac{1}{2}E_{P_0}[\phi(x)] + \frac{1}{2}E_{P_1}[1-\phi(x)]\nn\\&= \frac{1}{2}\int\phi(x) dP_0(x) + \frac{1}{2}\int\big(1-\phi(x)\big) dP_1(x),
\end{flalign}
where $P_0$ and $P_1$ denote the distributions under the null and alternative hypotheses, respectively. For the Bayesian robust hypothesis testing, the goal is to solve the following problem:
\begin{flalign}\label{eq:minmax}
\inf_\phi \sup_{P_0\in\mathcal{P}_0, P_1\in\mathcal{P}_1}P_E(\phi).
\end{flalign}
The results in this paper can be easily generalized to the case with non-equal priors.

Denote a sequence of independent and identically distributed (i.i.d.) samples by $x^n = (x_1, x_2, \cdots, x_n)$. The worst-case type-\uppercase\expandafter{\romannumeral1} error exponent $e_F$ and the worst-case type-\uppercase\expandafter{\romannumeral2} error exponent $e_M$ are defined as follows:
\begin{flalign}
&e_F(\phi) = \inf_{P_0\in\mathcal{P}_0}\lim_{n\rightarrow\infty} -\frac{1}{n} \log E_{P_0}[\phi(x^n)],\nn\\&
e_M(\phi) = \inf_{P_1\in\mathcal{P}_1}\lim_{n\rightarrow\infty} -\frac{1}{n} \log E_{P_1}[1-\phi(x^n)].
\end{flalign}

\begin{definition}
A test $\phi$ is said to be exponentially consistent if $e_F(\phi) >0$ and $e_M(\phi) >0$.
\end{definition}

\textit{2) Neyman-Pearson Setting}. In this paper, we focus on the asymptotic Neyman-Pearson setting, where the goal is to solve the following problem:
\begin{flalign}\label{eq:goal}
\sup_{\phi:P_F(\phi)\leq \alpha} \inf_{P_1\in\mathcal{P}_1}\lim_{n\rightarrow\infty} -\frac{1}{n} \log E_{P_1}[1-\phi(x^n)],
\end{flalign}
where $\alpha\in(0, 1]$ is a pre-specified constraint on the worst-case false alarm probability. Specifically, among the tests that satisfy the false alarm constraint $P_F(\phi) \leq \alpha$, we aim to find one that maximizes the worst-case type-\uppercase\expandafter{\romannumeral2} error exponent.

In this paper, any distributions $P_0, P_1$ are assumed to admit probability density functions (PDFs) $p_0, p_1$, since we can always choose a reference measure $\mu$ such that both $P_0$ and $P_1$ are absolutely continuous with respect to $\mu$. In general, $\mu$ can be chosen as $\mu = P_0 + P_1$. For the continuous distributions and the discrete distributions, $\mu$ can be chosen as the Lebesgue measure and the counting measure, respectively.

\section{Robust Hypothesis Testing Under Bayesian Setting}\label{sec:baysetting}
In this section, we focus on the Bayesian setting. We aim to solve the minimax problem for the average probability of error in \eqref{eq:minmax}.
	
\subsection{Finite-Alphabet Case}
Consider the case with a finite alphabet, i.e., $N \triangleq |\mathcal{X}|<\infty$. Let $\mathcal{X} = \{z_1, z_2,\cdots, z_N\}$. Then, $\hat{x}_{l,j}\in \{z_i\}_{i=1}^N \ \text{for}\ l = 0,1, j = 1, \cdots, m $. In this case, 
\begin{flalign}
P_E(\phi)=\sum_{i=1}^N (1-\phi_N(z_i))P_1^N(z_i) + \phi_N(z_i)P_0^N(z_i),
\end{flalign}
where we introduce the superscript $N$ on $P_0$ and $P_1$ to emphasize its dependence on $N$.  Therefore, \eqref{eq:minmax} can be written as
\begin{flalign}\label{eq:finiteminimax}
\inf_{\phi_N \in [0, 1]^{\otimes N}}\sup_{P_0^N \in \mathcal{P}_0, P_1^N \in \mathcal{P}_1}\sum_{i=1}^N (1-\phi_N(z_i))P_1^N(z_i) + \phi_N(z_i)P_0^N(z_i).
\end{flalign}
Note that \eqref{eq:finiteminimax} is a minimax problem. We then provide the following strong duality result for \eqref{eq:finiteminimax}, which is a finite-dimensional convex optimization problem.

\begin{lemma}\label{theorem:strongdual}
The minimax problem in \eqref{eq:finiteminimax} has the following strong dual formulation:  
	\begin{flalign}\label{eq:finitetract}
	\inf_{\phi_N\in [0, 1]^{\otimes N}, f_0, g_0, \alpha_j, \beta_j\in \mathbb{R}}& f_0 + g_0 + \frac{1}{m}\sum_{i=1}^m \sum_{j=1}^N \alpha_j k(z_j, \hat{x}_{1,i}) + \frac{1}{m}\sum_{i=1}^m \sum_{j=1}^{N} \beta_j k(z_j, \hat{x}_{0,i})\nn\\&+ \theta\Big\|\sum_{j=1}^N \alpha_j k(z_j, \cdot)\Big\|_{\mathcal{H}}+ \theta\Big\|\sum_{j=1}^{N} \beta_j k(z_j, \cdot)\Big\|_{\mathcal{H}}\nn\\
	\text{subject to}\ & 1-\phi_N(z_i) \leq f_0 + \sum_{j=1}^N \alpha_j k(z_j, z_i),\ \text{for}\ i = 1, \cdots, N\nn\\
	& \phi_N(z_i) \leq g_0 + \sum_{j=1}^{N} \beta_j k(z_j, z_i),\ \text{for}\ i = 1, \cdots, N\nn\\& 0\leq \phi_N(z_i) \leq 1,\ \text{for}\ i = 1, \cdots, N,
	\end{flalign}
which is a finite-dimensional convex optimization problem.
\end{lemma}
\begin{proof}
From the strong duality of kernel robust optimization \cite[Theorem 3.1]{zhu2021kernel}, we have that \eqref{eq:finiteminimax} has the following dual problem and the strong duality holds. 
\begin{flalign}\label{eq:finitewithf}
\inf_{\phi_N\in [0, 1]^{\otimes N}, f_0, g_0\in \mathbb{R}, f_1, g_1\in\mathcal{H}}& f_0 + g_0 + \frac{1}{m}\sum_{i=1}^m f_1(\hat{x}_{1,i}) + \frac{1}{m}\sum_{i=1}^m g_1(\hat{x}_{0,i})+ \theta\|f_1\|_{\mathcal{H}}+ \theta\|g_1\|_{\mathcal{H}}\nn\\
\text{subject to}\ & 1-\phi_N(z_i) \leq f_0 + f_1(z_i),\ \text{for}\ i = 1, \cdots, N\nn\\
& \phi_N(z_i) \leq g_0 + g_1(z_i),\ \text{for}\ i = 1, \cdots, N\nn\\& 0\leq \phi_N(z_i) \leq 1,\ \text{for}\ i = 1, \cdots, N,
\end{flalign}

From the robust representer theorem \cite{zhu2021kernel}, the functions $f_1, g_1$ admit the finite expansions $f_1(\cdot) = \sum_{j=1}^N \alpha_j k(z_j, \cdot)$ and $g_1(\cdot) = \sum_{j=1}^{N} \beta_j k(z_j, \cdot)$. Therefore, the optimization problem in \eqref{eq:finiteminimax} can be reformulated as a finite-dimensional convex optimization problem thus can be solved efficiently in practice. 	
\end{proof}
	
Note that \eqref{eq:finitetract} is a convex optimization problem with linear constraints, and thus can be solved using standard optimization tools \cite{cvx}. By solving \eqref{eq:finitetract}, we obtain the optimal robust test $\phi_N^*$ and can also find the optimal solutions $P_0^{*, N}, P_1^{*, N}$ for the inner problem in \eqref{eq:finiteminimax} by plugging $\phi_N^*$ back to \eqref{eq:finiteminimax}. When $\phi_N^*$ is known, \eqref{eq:finiteminimax} reduces to a finite-dimensional convex optimization problem and can be solved efficiently. In the following section, we also show that the results in the finite-alphabet case can be used to provide an asymptotically accurate approximation for the infinite-alphabet case.

\subsection{Infinite-Alphabet Case}
Consider the case where $\mathcal{X}$ is infinite. Then, \eqref{eq:finitetract} is infinite-dimensional, and is not directly solvable. To simplify the analysis of \eqref{eq:minmax}, we first interchange the $\sup$ and $\inf$ operators in \eqref{eq:minmax} based on the following proposition. Since the likelihood ratio test is optimal for the binary hypothesis testing problem, the inner problem can be solved by applying the likelihood ratio test. The original problem is then converted to a maximization problem. 
\begin{proposition}
The minimax problem in \eqref{eq:minmax} has the following reformulation:
\begin{flalign}\label{eq:sion}
&\inf_\phi \sup_{P_0\in\mathcal{P}_0, P_1\in\mathcal{P}_1}P_E(\phi) = \sup_{P_0\in\mathcal{P}_0, P_1\in\mathcal{P}_1}\inf_\phi P_E(\phi)\nn\\ =&\frac{1}{2} \sup_{P_0\in\mathcal{P}_0, P_1\in\mathcal{P}_1} \int \min\big\{p_0(x), p_1(x)\big\}dx.
\end{flalign}
\end{proposition}
\begin{proof}
The error probability $P_E(\phi)$ is continuous, real-valued and linear in $\phi$, $P_0$ and $P_1$. For any distributions $Q_1, Q_2 \in \mathcal{P}_l$, $l = 0, 1$, from the triangle inequality of MMD \cite{gretton2012kernel}, the convex combination $\lambda Q_1 + (1-\lambda) Q_2, 0<\lambda<1$, lies in $\mathcal{P}_l$. Therefore, the uncertainty set $\mathcal{P}_0$ and $\mathcal{P}_1$ are convex sets and $\mathcal{P}_0 \times\mathcal{P}_1$ is also convex. Denote by $\Phi$ the collection of all $\phi$. We have that $\Phi$ is the product of uncountably many compact sets of $[0, 1]$. Since $\mathcal{X}$ is compact, from the Tychonoff's theorem \cite{tychonoff1930, johnstone1981}, $\Phi$ is compact with respect to the product topology. Moreover, for any $\phi_1, \phi_2 \in \Phi$, the convex combination $\lambda\phi_1 + (1-\lambda)\phi_2,\ 0<\lambda<1$, also lies in $\Phi$. Therefore, $\Phi$ is convex. From the Sion's minimax theorem \cite{sion1958minimax}, we have that 
\begin{flalign}
&\inf_\phi \sup_{P_0\in\mathcal{P}_0, P_1\in\mathcal{P}_1}P_E(\phi)\nn\\ =& \sup_{P_0\in\mathcal{P}_0, P_1\in\mathcal{P}_1}\inf_\phi P_E(\phi) \nn\\=& \sup_{P_0\in\mathcal{P}_0, P_1\in\mathcal{P}_1} \frac{1}{2}\int \mathbb{I}_{\big\{\frac{p_1(x)}{p_0(x)}\geq 1\big\}}p_0(x)dx + \frac{1}{2} \int \mathbb{I}_{\big\{\frac{p_1(x)}{p_0(x)}< 1\big\}}p_1(x)dx \nn\\=& \frac{1}{2}\sup_{P_0\in\mathcal{P}_0, P_1\in\mathcal{P}_1} \int \min\big\{p_0(x), p_1(x)\big\}dx,
\end{flalign}
where $\mathbb{I}$ denotes the indicator function and the second equality is due to the fact that the likelihood ratio test is optimal for the binary hypothesis testing problem \cite{moulin2018statistical, kay1993fundamentals}.
\end{proof}
Observe that the problem in \eqref{eq:sion} is an infinite-dimensional optimization problem and the closed-form optimal solutions $P_0^*, P_1^*$ are difficult to derive. In the following, we propose a tractable approximation for the minimax error probability in \eqref{eq:sion}. With this approximation, the worst-case error probability in \eqref{eq:sion} can be quantified. The optimal solutions of this tractable approximation can be further used to design a robust test that generalizes to unseen samples.

Let $P$ be an arbitrary distribution supported on the whole space $\mathcal{X}$, and is absolutely continuous w.r.t. a uniform distribution on $\mathcal{X}$. Let $\{z_i\}^N_{i=1}$ be $N$ i.i.d. samples generated from $P$. We then propose the following approximation of \eqref{eq:sion} by restricting to distributions supported on the $N$ samples:
\begin{flalign}\label{eq:tractable}
\frac{1}{2}\sup_{P_0^N\in\mathcal{P}_0^N, P_1^N\in\mathcal{P}_1^N}\ \ &\sum_{i=1}^N \min\big\{P_0^N(z_i), P_1^N(z_i)\big\},
\end{flalign}
where $\mathcal{P}_l^N (l = 0, 1)$ denotes the collection of distributions that are supported on $\{z_i\}_{i=1}^N$ and satisfy $\big\|\mu_{P_l^N} - \mu_{\hat{Q}^l_m}\big\|_\mathcal{H} \leq \theta$.
We note that \eqref{eq:tractable} is a finite-dimensional convex optimization problem which can be solved by standard optimization tools. Let
\begin{flalign}\label{eq:definef}
&f(\mathcal{P}_0, \mathcal{P}_1) = \frac{1}{2} \sup_{P_0\in\mathcal{P}_0, P_1\in\mathcal{P}_1} \int \min\big\{p_0(x), p_1(x)\big\}dx,\nn\\& f(\mathcal{P}^N_0, \mathcal{P}^N_1) = \frac{1}{2}\sup_{P_0^N\in\mathcal{P}_0^N, P_1^N\in\mathcal{P}_1^N}\sum_{i=1}^N \min\big\{P_0^N(z_i), P_1^N(z_i)\big\}.
\end{flalign} 
Clearly, \eqref{eq:tractable} is a lower bound of \eqref{eq:sion}, i.e., $f(\mathcal{P}^N_0, \mathcal{P}^N_1)\leq f(\mathcal{P}_0, \mathcal{P}_1)$. The following theorem demonstrates that as $N\rightarrow\infty$, the value of \eqref{eq:tractable} converges to the value of \eqref{eq:sion} almost surely.
\begin{theorem}\label{theorem:tractable}
As $N\rightarrow\infty$, $f(\mathcal{P}^N_0, \mathcal{P}^N_1)$ converges to $f(\mathcal{P}_0, \mathcal{P}_1)$ almost surely.
\end{theorem}
	
Before we prove Theorem \ref{theorem:tractable}, we will first show that $\int\min\big\{p_0(x), p_1(x)\big\}dx$ is upper semi-continuous in $P_0, P_1$ with respect to the weak convergence in the following lemma. 
\begin{lemma}\label{lemma:uppersemi}
$\int\min\big\{p_0(x), p_1(x)\big\}dx$ is upper semi-continuous in $P_0, P_1$ with respect to the weak convergence.
\end{lemma}
\begin{proof}
The proof of Lemma \ref{lemma:uppersemi} can be found in Appendix \ref{sec:prooflemma1}.
\end{proof}
	
With Lemma \ref{lemma:uppersemi}, we are ready to prove Theorem \ref{theorem:tractable}.
\begin{proof}[Proof sketch]
We first show that the solutions to $\sup_{P_0\in\mathcal{P}_0, P_1\in\mathcal{P}_1}\int\min\big\{p_0(x), p_1(x)\big\}dx$ exist and let $P_0^*, P_1^*$ denote the optimal solutions. We then show that there exist distributions $P_0^s, P_1^s$ supported on $s$ samples converging weakly to $P_0^*, P_1^*$ respectively, as $s\rightarrow \infty$. Thirdly, for a fixed $s$, we show that there exist distributions $P_0^{s,N}, P_1^{s, N}$ supported on $\{z_i\}_{i=1}^N$ converging weakly to $P_0^s, P_1^s$ almost surely as $N\rightarrow \infty$ and 
\begin{flalign}
\sum_{i=1}^N\min\big\{P_0^{s, N}(z_i), P_1^{s, N}(z_i)\big\} \geq \int\min\{p_0^*(x), p_1^*(x)\}dx.
\end{flalign} 
Moreover, we show that for any $\epsilon>0$, there exists large $s$ and $N$ such that the MMD between $P_0^{s,N}, P_1^{s, N}$ and $\hat{Q}^0_m, \hat{Q}^1_m$ can be bounded by $\theta+\epsilon$. Finally, by letting $\epsilon\rightarrow0$, we prove the convergence result in Theorem \ref{theorem:tractable}.

The full proof of Theorem \ref{theorem:tractable} can be found in Appendix \ref{sec:pfthm2}.
\end{proof}

Though the optimal solutions  $P_0^*, P_1^*$ for \eqref{eq:sion} is difficult to derive, \eqref{eq:tractable} provides a lower bound on the worst-case error probability, and is asymptotically accurate as $N\rightarrow \infty$.

\subsection{Convergence Rate}
In this section, we characterize the approximation error in Theorem \ref{theorem:tractable} in terms of $N$ if radial basis function (RBF) kernels \cite{vert2004primer} are used, i.e., $k(x, y) = \exp\big(-\frac{\|x-y\|^\rho_2}{2\sigma^2}\big)$, where $\rho >0$ is some constant. For example, when $\rho = 1$, $k(x, y)$ is the Laplacian kernel (exponential kernel); and when $\rho = 2$, $k(x, y)$ is the Gaussian kernel. A $\delta$-net of $\mathcal{X}$ is a set of points $\{z_i\}_{i=1}^N$ in $\mathcal{X}$ such that for any $z\in\mathcal{X}$, there exists some $z_i$ that satisfies $\|z-z_i\|_2^\rho\leq \delta$. From classic covering number results, we can construct a $\delta$-net where $\delta = \frac{\max_{z, z^\prime\in\mathcal{X}}\|z-z^\prime\|_2^\rho}{N^{\frac{1}{d}}}$ with the number of $N$ points.
Following the similar idea as in \cite{magesh2022robust}, we use $\{z_i\}_{i=1}^N$ to construct the support sample set. There exists a partition $\mathcal{A}_N = \{\mathcal{A}_N^1, \mathcal{A}_N^2, \cdots, \mathcal{A}_N^N\}$ based on the $\delta$-net such that $z_i\in \mathcal{A}_N^i$ and $\max_{z\in\mathcal{A}_N^i}\|z-z_i\|_2^\rho \leq \delta$.
We rewrite $f(\mathcal{P}_0, \mathcal{P}_1), f(\mathcal{P}_0^N, \mathcal{P}_1^N)$ in \eqref{eq:definef} as a function of $\theta$ and define
\begin{flalign}\label{eq:defineg}
g(\theta) = \sup_{\substack{P_0\in\mathcal{P}:\big\|\mu_{P_0} - \mu_{\hat{Q}^0_m}\big\|_\mathcal{H}\leq \theta \\P_1\in\mathcal{P}:\big\|\mu_{P_1} - \mu_{\hat{Q}^1_m}\big\|_\mathcal{H}\leq \theta}}\int\min\big\{p_0(x), p_1(x)\big\}dx,
\end{flalign} 
and
\begin{flalign}\label{eq:definegn}
g_N(\theta) = \sup_{\substack{P_0^N\in\mathcal{P}:\big\|\mu_{P_0^N} - \mu_{\hat{Q}^0_m}\big\|_\mathcal{H}\leq \theta \\ P_1^N\in\mathcal{P}:\big\|\mu_{P_1^N} - \mu_{\hat{Q}^1_m}\big\|_\mathcal{H}\leq \theta\\ P_0^N, P_1^N\ \text{are supported on}\ \{z_i\}_{i=1}^N}}\sum_{i=1}^N\min\big\{P_0^N(z_i), P_1^N(z_i)\big\}.
\end{flalign}
We note that for RBF kernels, $0\leq k(x,y)\leq 1$. Therefore, for any $P_0\in\mathcal{P}_0$, we have that $\big\|\mu_{P_0} - \mu_{\hat{Q}^0_m}\big\|_\mathcal{H} = \Big(E_{x\sim P_0, x^\prime\sim P_0}[k(x,x^\prime)] + E_{y\sim \hat{Q}^0_m, y^\prime\sim \hat{Q}^0_m}[k(y,y^\prime)] - 2E_{x\sim P_0, y\sim \hat{Q}^0_m}[k(x,y)]\Big)^{1/2}\leq \sqrt{2}.$ Similarly, we have that $\big\|\mu_{P_1} - \mu_{\hat{Q}^1_m}\big\|_\mathcal{H}\leq \sqrt{2}$ for any $P_1\in\mathcal{P}_1$. Therefore, it suffices to consider $\theta\in(0, \sqrt{2}]$.
We then have the following theorem that bounds the approximation error similar as in \cite{magesh2022robust}.
\begin{theorem}\label{theorem:rate}
	Let $\epsilon = \sqrt{2-2\exp\Big(-\frac{\delta}{2\sigma^2}\Big)}$. For any $\theta\in(0, \sqrt{2}]$, the approximation error satisfies $|g(\theta)-g_N(\theta)|\leq L\epsilon$, where $L$ is some constant.
\end{theorem}
\begin{proof}
Denote by $P_0^N, P_1^N$ discrete distributions with $P_0^N(z_i) = P_0^*(\mathcal{A}_N^i), P_1^N(z_i) = P_1^*(\mathcal{A}_N^i)$. Consider the MMD between $P_0^N$ and $P_0^*$, we then have that
\begin{flalign}\label{eq:errorn}
&\|\mu_{P_0^N}-\mu_{P_0^*}\|_{\mathcal{H}} \nn\\&= 
\sup_{h:\|h\|_\mathcal{H}\leq 1}\int hdP_0^N - \int hdP_0^*\nn\\&=\sup_{h:\|h\|_\mathcal{H}\leq 1}\sum_{i=1}^N\int_{\mathcal{A}_N^i}\big(h-h(z_i)\big)dP_0^*\nn\\&\leq\sup_{h:\|h\|_\mathcal{H}\leq 1}\sum_{i=1}^N\int_{\mathcal{A}_N^i}\max_{z\in\mathcal{A}_N^i}\big(h(z)-h(z_i)\big)dP_0^*\nn\\& \overset{(a)}{=}\sup_{h:\|h\|_\mathcal{H}\leq 1}\sum_{i=1}^N\int_{\mathcal{A}_N^i}\max_{z\in\mathcal{A}_N^i}\langle h, k(z,\cdot) - k(z_i, \cdot)\rangle_{\mathcal{H}} dP_0^*\nn\\ &\overset{(b)}{\leq}\sup_{h:\|h\|_\mathcal{H}\leq 1}\sum_{i=1}^N\int_{\mathcal{A}_N^i}\max_{z\in\mathcal{A}_N^i}\|h\|_{\mathcal{H}} \|k(z, \cdot)-k(z_j, \cdot)\|_{\mathcal{H}}dP_0^*\nn\\ &\leq \sup_{h:\|h\|_\mathcal{H}\leq 1}\sum_{i=1}^N\int_{\mathcal{A}_N^i} \|h\|_{\mathcal{H}}\max_{z \in\mathcal{A}_N^i} \sqrt{k(z, z)+k(z_i, z_i)-2k(z,z_i)}dP_0^* \nn\\& =\sum_{i=1}^N\int_{\mathcal{A}_N^i}\max_{z\in\mathcal{A}_N^i} \sqrt{k(z, z)+k(z_i, z_i)-2k(z, z_i)}dP_0^* \nn\\&= \sum_{i=1}^N\int_{\mathcal{A}_N^i} \sqrt{2-2\min_{z\in\mathcal{A}_N^i} k(z, z_i)}dP_0^*\nn\\&\overset{(c)}{\leq} \sqrt{2-2\exp\Big(-\frac{\delta}{2\sigma^2}\Big)},
\end{flalign}
where $(a)$ is from the reproducing property of the kernel, $(b)$ is from the Cauchy-Schwartz inequality and $(c)$ is due to the fact that $\min_{z\in\mathcal{A}_N^i} k(z, z_i)\geq \exp\Big(-\frac{\delta}{2\sigma^2}\Big)$.
Let $\epsilon = \sqrt{2-2\exp\Big(-\frac{\delta}{2\sigma^2}\Big)}$. We then have that $\|\mu_{P_0^N} - \mu_{\hat{Q}^0_m}\|_{\mathcal{H}}\leq \theta+\epsilon$ by the triangle inequality. Therefore, $P_0^N$ lies in the uncertainty set centered around $\hat{Q}^0_m$ with radius $\theta + \epsilon$. Similarly, for $P_1^N$ and $P_1^*$, we have the same result that $P_1^N$ lies in the uncertainty set centered around $\hat{Q}^1_m$ with radius $\theta + \epsilon$. From Jensen's inequality \cite{jensen1906inequality}, we have that 
\begin{flalign}
&\int\min\big\{p_0^*(x), p_1^*(x)\big\}dx \leq \sum_{i=1}^N \min\big\{P_0^N(z_i), P_1^N(z_i)\big\}.
\end{flalign}
Therefore, $g_N(\theta+\epsilon) \geq g(\theta)$.
We then have that 
\begin{flalign}
|g(\theta)-g_N(\theta)|&\leq  |g_N(\theta+\epsilon)-g(\theta)| + |g_N(\theta+\epsilon) - g_N(\theta)|\nn\\&\leq |g(\theta+\epsilon)-g(\theta)| + |g_N(\theta+\epsilon) - g_N(\theta)|,
\end{flalign}
where the first inequality is due to the fact that $g_N(\theta+\epsilon) \geq g(\theta)$ and $g(\theta + \epsilon)\geq g_N(\theta + \epsilon)$. 
We will then bound the first term $|g(\theta+\epsilon)-g(\theta)|$. Since $g(\theta)$ is concave when $\theta\in(0, \infty)$ (see \eqref{eq:concaveg} in Appendix \ref{sec:pfthm2} for the proof), there exists $L_1$ such that $g(\cdot)$ is $L_1$-Lipschitz on $[\theta, \sqrt{2}]$ \cite{roberts1974lipschitz}. Therefore, $|g(\theta+\epsilon)-g(\theta)|\leq L_1\epsilon$. Similarly, there exists $L_2$ such that $|g_N(\theta+\epsilon) - g_N(\theta)|\leq L_2\epsilon$. Therefore, $|g(\theta)-g_N(\theta)|\leq (L_1 + L_2)\epsilon$ for any $\theta\in(0, \sqrt{2}]$. Note that $L_1$ and $L_2$ may depend on $\theta$. 
\end{proof}
This result characterizes the convergence rate of the approximation error with respect to the support sample size $N$. In practice, we can choose a proper $N$ to control the approximation error. This result also reveals the relation between the data dimension $d$ and the number of support sample $N$. In high dimensional setting, we need a larger $N$ to achieve the same approximation error as in the low dimensional setting. Note that $N$ is the number of artificially generated samples, and therefore we could generate as many as we like at the price of increased computational cost for solving \eqref{eq:tractable}.	
	
\subsection{Robust Test via Kernel Smoothing}
	
	

The optimal solutions $P_0^*, P_1^*$ of \eqref{eq:sion}, and thus the likelihood ratio test between $P_0^*$ and $P_1^*$ are difficult to derive. Note that $P_0^{*, N}, P_1^{*, N}$ are optimal solutions to \eqref{eq:tractable}. The following proposition shows that the sequence $\{P^{*,N}_0, P^{*,N}_1\}_{N=1}^{\infty}$ converges weakly to an optimal solution of \eqref{eq:sion}, i.e., for all bounded and continuous functions $h$, $\lim_{N\rightarrow\infty}E_{P^{*,N}_0}[h] = E_{P_0^*}[h]$ and $\lim_{N\rightarrow\infty}E_{P^{*,N}_1}[h] = E_{P_1^*}[h]$. The fact that $P_0^{*, N}, P_1^{*, N}$ are reasonable approximations of $P_0^*, P_1^*$ as $N\rightarrow \infty$ further motivates our kernel smoothing method to design a robust test that generalizes to the entire alphabet in this section.
\begin{proposition}\label{proposition:weakcon}
	The sequence $\{P^{*,N}_0, P^{*,N}_1\}_{N=1}^{\infty}$ converges weakly to an optimal solution of \eqref{eq:sion}. 
\end{proposition}
\begin{proof}
Observe that for any $N$, $\{P^{*,N}_0, P^{*,N}_1\}$ lies in the compact set $\mathcal{P}_0 \times \mathcal{P}_1$. Assume $\{P^{*,N}_0, P^{*,N}_1\}$ does not converge weakly to the optimal solution $\{P_0^*, P_1^*\}$, then there exists a subsequence of  $\{P^{*,N}_0, P^{*,N}_1\}_{N=1}^\infty$ converges weakly but not to $\{P_0^*, P_1^*\}$ \cite[Chapter 5.1.1]{shapiro2021lectures}. 
Denote the sequence by $\{P_0^{*, N(t)}, P_1^{*, N(t)}\}_{t=1}^\infty$. Assume $\{P_0^{*, N(t)}, P_1^{*, N(t)}\}_{t=1}^\infty$ converges weakly to $\{P_0^\prime, P_1^\prime\}$. Since $\{P_0^\prime, P_1^\prime\}$ is not an optimal solution to \eqref{eq:sion}, we have that 
\begin{flalign}
\int\min\{p_0^\prime(x), p_1^\prime(x)\}dx < \int\min\{p_0^*(x), p_1^*(x)\}dx.
\end{flalign}
We then have that 
\begin{flalign}\label{eq:subseq}
	&\int\min\{p_0^*(x), p_1^*(x)\}dx\nn\\&= \lim_{t\rightarrow\infty}\sum_{i=1}^N\min\{P_0^{*, N(t)}(z_i), P_1^{*, N(t)}(z_i)\} \nn\\&\leq\int\min\{p_0^\prime(x), p_1^\prime(x)\}dx,
\end{flalign}
where the equality is from Theorem \ref{theorem:tractable} and the inequality is due to the upper semi-continuity of $\int\min\{p_0(x), p_1(x)\}dx$ in Lemma \ref{lemma:uppersemi}. This leads to a contradiction. Therefore, $\{P^{*,N}_0, P^{*,N}_1\}_{N=1}^\infty$ converges weakly an optimal solution of  \eqref{eq:sion}.
\end{proof}

Note that $P^{*,N}_0, P^{*,N}_1$ are convex combinations of Dirac measures. We then extend them to the whole space via kernel smoothing to approximate $P_0^*, P_1^*$, i.e.,
\begin{flalign}
&\widetilde{P}_0^*(x) = \sum_{i=1}^NP_0^{*,N}(z_i) k(x, z_i),\nn\\&\widetilde{P}_1^*(x) = \sum_{i=1}^NP_1^{*,N}(z_i) k(x, z_i).
\end{flalign}
The kernel functions have various choices. For example, the Gaussian kernel with bandwidth parameter $\sigma$: $k(x,y) = \frac{1}{\sqrt{2}\sigma} \exp\big(-\frac{\|x-y\|^2}{2\sigma^2}\big)$. After kernel smoothing, we define the likelihood ratio test $\tilde{\phi}$ between $\widetilde{P}_1^*(x)$ and $\widetilde{P}_0^*(x)$ over the whole space $\mathcal{X}$ as follows to approximate the optimal test:
\begin{flalign}\label{eq:smoothtest}
\tilde{\phi}(x)=\left\{\begin{array}{ll}
1, &\text{ if } \log\frac{\widetilde{P}_1^*(x)}{\widetilde{P}_0^*(x)}\geq 0 \\
0, &\text{ if } \log\frac{\widetilde{P}_1^*(x)}{\widetilde{P}_0^*(x)} < 0.
\end{array}\right.
\end{flalign}
When the testing sample size is $n$, after solving $P_0^{*,N}, P_1^{*,N}$, the computational complexity for implementing $\tilde{\phi}$ is $\mathcal{O}(nN)$. The numerical results in Section \ref{sec:simulation} show that $\tilde{\phi}$ performs well in practice, and is robust to model uncertainty.

\subsection{A Direct Robust Kernel Test}
In this section, we consider the problem of testing a sequence of samples $x^n$, where $n$ is the sample size. We propose a direct robust kernel test and further show that it is exponentially consistent as $n\rightarrow \infty$ under the Bayesian setting.

Motivated by the facts that the MMD can be used to measure the distance between distributions when the kernel $k$ is characteristic, we propose a direct robust kernel test as follows
\begin{flalign}\label{eq:bayesiantest}
\phi_B(x^n)=\left\{\begin{array}{ll}
1, &\text{ if } S(x^n)\geq \gamma \\
0, &\text{ if } S(x^n) < \gamma,
\end{array}\right.
\end{flalign}
where
\begin{flalign}
S(x^n) = \inf_{P\in\mathcal{P}_0}\big\|\mu_{\hat{P}_n} - \mu_{P}\big\|_\mathcal{H} - \inf_{P\in\mathcal{P}_1}\big\|\mu_{\hat{P}_n} - \mu_{P}\big\|_\mathcal{H},
\end{flalign} 
and $\gamma$ is a pre-specified threshold. In the construction of $S(x^n)$, we use ``inf" to tackle the uncertainty of distributions and compare the closest distance between the empirical distribution of samples and the two uncertainty sets. The test statistic involves two infinite-dimensional optimization problems, and thus is difficult to solve in general. In the following proposition, we show that it can actually be solved analytically in closed-form, and the computational complexity is $\mathcal O(m^2+n^2)$.

\begin{proposition}\label{proposition: tractable}
For any $l = 0, 1$, if $\big\|\mu_{\hat{P}_n} - \mu_{\hat{Q}^l_m}\big\|_\mathcal{H} > \theta$, then
\begin{flalign}
&\inf_{P\in\mathcal{P}_l}\big\|\mu_{\hat{P}_n} - \mu_{P}\big\|_\mathcal{H} = \big\|\mu_{\hat{P}_n} - \mu_{\hat{Q}^l_m}\big\|_\mathcal{H} - \theta \nn\\&= \Big(\frac{1}{n^2}\sum_{i=1}^n\sum_{j=1}^nk(x_i, x_j) + \frac{1}{m^2}\sum_{i=1}^m\sum_{j=1}^mk(\hat{x}_{l,i}, \hat{x}_{l,j}) - \frac{2}{nm}\sum_{i=1}^n\sum_{j=1}^mk(x_i, \hat{x}_{l,j})\Big)^{1/2} - \theta;
\end{flalign}
and otherwise, $\inf_{P\in\mathcal{P}_l}\big\|\mu_{\hat{P}_n} - \mu_{P}\big\|_\mathcal{H} = 0$.
\end{proposition}
\begin{proof}
The proof of Proposition \ref{proposition: tractable} can be found in Appendix \ref{sec:pfcompute}.
\end{proof}

In the following theorem, we show that with a proper choice of $\gamma$, $\phi_B$ is exponentially consistent.
\begin{theorem}\label{theorem:bayesopt}
1) If $\gamma\in \Big(-\big\|\mu_{\hat{Q}_m^0} - \mu_{\hat{Q}_m^1}\big\|_\mathcal{H} + 2\theta, \big\|\mu_{\hat{Q}_m^0} - \mu_{\hat{Q}_m^1}\big\|_\mathcal{H} - 2\theta\Big)$, $\phi_B$ is exponentially consistent.\\
2) $\phi_B$ can be equivalently written as
\begin{flalign}
\phi_B^\prime(x^n)=\left\{\begin{array}{ll}
1, &\text{if } \big\|\mu_{\hat{P}_n} - \mu_{\hat{Q}_m^0}\big\|_\mathcal{H} - \big\|\mu_{\hat{P}_n} - \mu_{\hat{Q}_m^1}\big\|_\mathcal{H}\geq \gamma \\
0, &\text{if } \big\|\mu_{\hat{P}_n} - \mu_{\hat{Q}_m^0}\big\|_\mathcal{H} - \big\|\mu_{\hat{P}_n} - \mu_{\hat{Q}_m^1}\big\|_\mathcal{H} < \gamma,
\end{array}\right.
\end{flalign}
and its computational complexity is $\mathcal O\big(m^2+n^2\big)$.
\end{theorem}
\begin{proof}
The proof of Theorem \ref{theorem:bayesopt} can be found in Appendix \ref{sec:pfdirect}.
\end{proof}
It can be seen that the direct test for robust hypothesis testing naturally reduces to comparing the MMD distance between the empirical distribution of samples and two centers of uncertainty sets which is computationally efficient.	
The exponential consistency of $\phi_B$ implies that the error probabilities decay exponentially fast with the sample size $n$. In practice, we can choose a proper threshold to balance the trade-off between the two types of errors.

The error exponent in Theorem \ref{theorem:bayesopt} is in an asymptotic sense, and is in the form of an optimization problem without a closed-form solution. In the following proposition, we consider a special case with $\gamma = 0$ and derive the closed-form upper bound of the worst-case error probabilities. 
\begin{proposition}\label{proposition: worstbound}
Set $\gamma = 0$ in \eqref{eq:bayesiantest}. Then, the worst-case type-\uppercase\expandafter{\romannumeral1} and type-\uppercase\expandafter{\romannumeral2} errors can be bounded as follows,
\begin{flalign}\label{eq:type1bound}
&\sup_{P_0\in\mathcal{P}_0}E_{P_0}[\phi_B(x^n)]\leq \exp\Bigg(-\frac{n\Big(\big\|\mu_{\hat{Q}_m^1} - \mu_{\hat{Q}_m^0}\big\|_\mathcal{H}^2-2\theta\big\|\mu_{\hat{Q}_m^1} - \mu_{\hat{Q}_m^0}\big\|_\mathcal{H}\Big)^2}{8K^2}\Bigg)
\end{flalign}
and 
\begin{flalign}\label{eq:type2bound}
&\sup_{P_1\in\mathcal{P}_1}E_{P_1}[1-\phi_B(x^n)]\leq \exp\Bigg(-\frac{n\Big(\big\|\mu_{\hat{Q}_m^1} - \mu_{\hat{Q}_m^0}\big\|_\mathcal{H}^2-2\theta\big\|\mu_{\hat{Q}_m^1} - \mu_{\hat{Q}_m^0}\big\|_\mathcal{H}\Big)^2}{8K^2}\Bigg).
\end{flalign}
\end{proposition}
\begin{proof}
	The proof of Proposition \ref{proposition: worstbound} can be found in Appendix \ref{sec:pfworstbound}.
\end{proof}

In Proposition \ref{proposition: worstbound}, we provide an upper bound on the worst-case error probability of $\phi_B$ when $\gamma = 0$. It can be seen that the error probabilities decay exponentially fast with an exponent of $\Big(\big\|\mu_{\hat{Q}_m^1} - \mu_{\hat{Q}_m^0}\big\|_\mathcal{H}^2-2\theta\big\|\mu_{\hat{Q}_m^1} - \mu_{\hat{Q}_m^0}\big\|_\mathcal{H}\Big)^2/8K^2$, which validates the fact that $\phi_B$ is exponentially consistent. Moreover, the decay rate is a function of the radius $\theta$ and the MMD distance between centers of two uncertainty sets. When the centers of two uncertainty sets are fixed, the upper bound on the error probabilities will increase with the radius $\theta$. Proposition \ref{proposition: worstbound} provides a closed-form non-asymptotic upper bound on the worst-case error probability. In practice, this upper bound can be used to evaluate the worst-case risk of implementing $\phi_B$ for a finite sample size $n$. Moreover, combining Theorem \ref{theorem:tractable} and Theorem \ref{theorem:bayesopt}, the performance gap between $\phi_B$ and the optimal test can be approximated.

\section{Robust Hypothesis Testing under Neyman-Pearson Setting}\label{sec:npsetting}
In this section, we focus on the Neyman-Pearson setting. We propose a robust kernel test, and show that it is asymptotically optimal under the Neyman-Pearson setting. The results in this section also hold for $\mathcal{X} = \mathbb{R}^d$.

\subsection{Universal Upper Bound on the Worst-Case Error Exponent}\label{sec:upperbound}
In this section, we derive the universal upper bound on the error exponent for the problem in \eqref{eq:goal}. The following proposition is a robust version of the Chernoff-Stein lemma \cite{cover2006information, dembo2009large}.
\begin{proposition}\label{proposition:upper}
Consider the robust hypothesis testing problem in \eqref{eq:goal}, we have that 
\begin{flalign}
&\sup_{\phi:P_F(\phi)\leq \alpha} \inf_{P_1\in\mathcal{P}_1}\lim_{n\rightarrow\infty} -\frac{1}{n} \log E_{P_1}[1-\phi(x^n)]\nn\\&\leq \inf_{P_0\in\mathcal{P}_0, P_1\in\mathcal{P}_1}D(P_0\|P_1). 
\end{flalign}
\end{proposition}
\begin{proof}
	For any $P_0\in\mathcal{P}_0, P_1 \in \mathcal{P}_1$, from the Chernoff-Stein lemma \cite{cover2006information, dembo2009large}, we have that 
	\begin{flalign}
	\sup_{\phi:E_{P_0}[\phi(x^n)]\leq \alpha}\lim_{n\rightarrow\infty} -\frac{1}{n} \log E_{P_1}[1-\phi(x^n)] = D(P_0\|P_1).
	\end{flalign}
	Since $P_F(\phi) \triangleq \sup_{P_0\in\mathcal{P}_0}E_{P_0}[\phi(x^n)]$, we have that $\{\phi:P_F(\phi)\leq \alpha\}\subseteq\{\phi:E_{P_0}[\phi(x^n)]\leq \alpha\}$. Therefore, for any $P_0\in\mathcal{P}_0, P_1 \in \mathcal{P}_1$, we have that 
	\begin{flalign}\label{eq:random}
	&\sup_{\phi:P_F(\phi)\leq \alpha} \inf_{P_1\in\mathcal{P}_1}\lim_{n\rightarrow\infty} -\frac{1}{n} \log E_{P_1}[1-\phi(x^n)]\nn\\&\leq\sup_{\phi:E_{P_0}[\phi(x^n)]\leq \alpha}\lim_{n\rightarrow\infty} -\frac{1}{n} \log E_{P_1}[1-\phi(x^n)]\nn\\ &= D(P_0\|P_1).
	\end{flalign}
	The solutions to $\inf_{P_0\in\mathcal{P}_0, P_1\in\mathcal{P}_1}D(P_0\|P_1)$ exist since $D(P_0\|P_1)$ is lower semi-continuous and lower semi-continuous functions attain its infimum on a compact set. Since \eqref{eq:random} holds for any $P_0\in\mathcal{P}_0, P_1 \in \mathcal{P}_1$, we then have that
	\begin{flalign}
	&\sup_{\phi:P_F(\phi)\leq \alpha} \inf_{P_1\in\mathcal{P}_1}\lim_{n\rightarrow\infty} -\frac{1}{n} \log E_{P_1}[1-\phi(x^n)]\nn\\&\leq \inf_{P_0\in\mathcal{P}_0, P_1\in\mathcal{P}_1}D(P_0\|P_1).
	\end{flalign}
\end{proof}

Proposition \ref{proposition:upper} implies that for any test, the achievable error exponent is no better than $\inf_{P_0\in\mathcal{P}_0, P_1\in\mathcal{P}_1}D(P_0\|P_1)$. This theorem also applies to  robust hypothesis testing problems with different uncertainty sets.

\subsection{Asymptotically Optimal Robust Kernel Test}\label{sec:test}
In this section, we propose a robust kernel test for the problem in \eqref{eq:goal}, and further prove that it is asymptotically optimal.

Motivated by the fact that when the kernel $k$ is characteristic, the MMD is a metric and can be used to measure the distance between distributions and the kernel test is asymptotically optimal under the Neyman-Pearson setting for the two-sample test problem \cite{zhu2021asymptotically}, we design our robust kernel test as follows:
\begin{flalign}\label{eq:kerneltest}
\phi_N(x^n)=\left\{\begin{array}{ll}
1, &\text{ if } \inf_{P\in\mathcal{P}_0}\big\|\mu_{\hat{P}_n} - \mu_{P}\big\|_\mathcal{H} > \gamma_n \\
0, &\text{ if } \inf_{P\in\mathcal{P}_0}\big\|\mu_{\hat{P}_n} - \mu_{P}\big\|_\mathcal{H}\leq \gamma_n,
\end{array}\right.
\end{flalign}
where $\gamma_n= \sqrt{2K/n}\big(1+\sqrt{-\log\alpha}\big)$ is chosen to satisfy the false alarm constraint that $P_F(\phi_N)\leq \alpha$, and $\hat{P}_n$ is the empirical distribution of $x^n$. We use ``inf" in the test statistic to tackle the uncertainty of distributions. A heuristic explanation of our test is that we use the closest distance between the empirical distribution of the test samples and the uncertainty set $\mathcal{P}_0$.
Our test statistic does not depend on $\mathcal P_1$, but later we will show that it is asymptotically optimal under the Neyman-Pearson setting, i.e., solves the problem in \eqref{eq:goal}.


From Proposition \ref{proposition: tractable}, we have that the test statistic of our robust kernel test can be solved analytically in closed-form with a computational complexity of $\mathcal{O}(m^2+n^2)$. We then show that the kernel robust test in \eqref{eq:kerneltest} is asymptotically optimal for the problem in \eqref{eq:goal} in the following theorem, i.e., it achieves the universal upper bound on the worst-case error exponent in Proposition \ref{proposition:upper}. 
\begin{theorem}\label{theorem:error}
The robust kernel test in \eqref{eq:kerneltest} is asymptotically optimal under Neyman-Pearson setting:\\
1) under $H_0$, 
\begin{flalign}
\sup_{P_0\in\mathcal{P}_0}E_{P_0}[\phi_N(x^n)] = \sup_{P_0\in\mathcal{P}_0}P_0 \Big(\inf_{P\in\mathcal{P}_0}\big\|\mu_{\hat{P}_n} - \mu_{P}\big\|_\mathcal{H} > \gamma_n\Big)\leq \alpha;
\end{flalign}
and 2) under $H_1$, 
\begin{flalign}
&\inf_{P_1\in\mathcal{P}_1}\lim_{n\rightarrow\infty}-\frac{1}{n}\log E_{P_1}[1-\phi_N(x^n)]\nn\\& = \inf_{P_1\in\mathcal{P}_1}\lim_{n\rightarrow\infty}-\frac{1}{n}\log P_1 \Big(\inf_{P\in\mathcal{P}_0}\big\|\mu_{\hat{P}_n} - \mu_{P}\big\|_\mathcal{H}\leq \gamma_n\Big)\nn\\&=\sup_{\phi:P_F(\phi)\leq \alpha} \inf_{P_1\in\mathcal{P}_1}\lim_{n\rightarrow\infty} -\frac{1}{n} \log E_{P_1}[1-\phi(x^n)]\nn\\&= \inf_{P_0\in\mathcal{P}_0,P_1\in\mathcal{P}_1} D(P_0\|P_1).
\end{flalign}	
\end{theorem}
\begin{proof}
The proof of Theorem \ref{theorem:error} can be found in Appendix \ref{sec:pfnperror}.
\end{proof}
The optimality result for the kernel robust test \eqref{eq:kerneltest} in Theorem \ref{theorem:error} holds for general robust hypothesis testing  problems, i.e.,  it applies to robust hypothesis testing problems defined using different uncertainty sets $\mathcal{P}_0, \mathcal{P}_1$ and using any arbitrary nominal distributions. However, to solve the optimization problem in the test statistic $\inf_{P\in\mathcal{P}_0}\big\|\mu_{\hat{P}_n} - \mu_{P}\big\|_\mathcal{H}$ for any arbitrary uncertainty set $\mathcal P_0$, it may not always be tractable, since it is an infinite-dimensional problem.

\section{Simulation Results}\label{sec:simulation}
In this section, we provide some numerical results to demonstrate the performance of our proposed tests.
	
\subsection{A Toy Example}
We first provide a toy example to visualize the impact of our kernel robust framework. Assume the whole space $\mathcal{X} = \{1, 2, 3, 4, 5\}$. Under hypothesis $H_0$, the training samples are 1, 2, 3. Under hypothesis $H_1$, the training samples are 3, 4, 5. The radius is set to be $0.2$. We choose a Gaussian kernel. The bandwidth for the Gaussian kernel is chosen using the medium heuristic\cite{gretton2012kernel}.
We plot the empirical distributions and the least favorable distributions (LFDs). The supports of the empirical distributions overlap only at $x=3$. Comparing the empirical distributions and the LFDs, it can be seen from Fig.~\ref{toy} that under $H_0$, part of the probability mass is transported from $\{1, 2\}$ to $\{4, 5\}$, under $H_1$, part of the probability mass is transported from $\{4, 5\}$ to $\{1, 2\}$. Therefore, the LFDs are more difficult to distinguish than the empirical distributions.

\begin{figure}[htb]
	\centering 
	\subfigure[Empirical Distributions]{
		\label{fig:empirical}
		\includegraphics[width=0.48\linewidth]{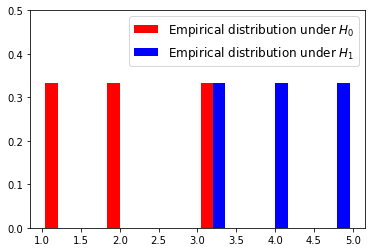}}
	\subfigure[Least Favorable Distributions]{
		\label{fig:lfds}
		\includegraphics[width=0.48\linewidth]{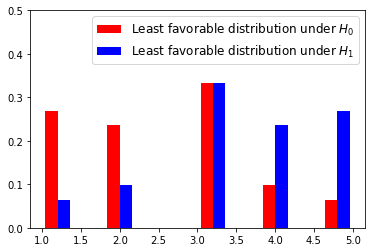}}
	\caption{Comparison of the Empirical Distributions and the Least Favorable Distributions on a Toy Example.}
	\label{toy}
\end{figure}

\subsection{Exponential Consistency of the Tests}
In this section, we validate the exponential consistency of the proposed tests. We use 40 samples from $\mathcal{N}(0, \mathbf{I})$ and 40 samples from $\mathcal{N}(0.22\mathbf{e}, \mathbf{I})$ to construct the uncertainty sets under $H_0$ and $H_1$ respectively, where $\mathbf{e}$ is a vector with all entries equal to 1, and $\mathbf{I}$ is the identity matrix. The data dimension is 20. The radii are chosen such that the uncertainty sets do not overlap. For the kernel smoothing robust test, we use training samples as the support of the finite-dimensional robust optimization problem in \eqref{eq:definef}. When testing the batch samples, we take the sum of the log-likelihood ratio for each sample and compare it with a threshold.
We choose the Gaussian kernel and the bandwidth parameter is chosen using cross-validation. We use the data-generating distributions to evaluate the performance of the two tests. We plot the log of the error probability under the Bayesian setting as a function of testing sample size $n$. It can be seen from Fig.~1 that with the increasing of sample size $n$, the error probabilities of the direct robust test and kernel smoothing robust test decay exponentially fast, which validates the theoretical result that the direct robust test is exponentially consistent. Moreover, our kernel smoothing robust test has a better performance than the direct robust kernel test. 

\begin{figure}[htbp]
	\centering
	\includegraphics[width=3.2in]{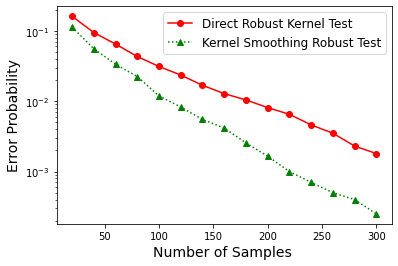}
	\caption{Exponential Consistency of the Tests.}
	\label{fig:exponential}
\end{figure}

\subsection{Comparison of the Performance}\label{sec:numBay}
In this section, we compare our tests with the Wasserstein robust test \cite{gao2018wasserstein, xie2021robust}.

We first compare the performance using synthetic data. We use $20$ samples from 
$\mathcal{N}(\mathbf{0}, \mathbf{I})$ and $20$ samples from $\mathcal{N}(\mathbf{e}, \mathbf{I})$ to construct the uncertainty sets under $H_0$ and $H_1$ respectively. The data dimension is 4. We use a Gaussian kernel and the bandwidth parameter is chosen using the cross-validation. For a fair comparison, and due to the difficulty of obtaining the coefficients in the Wasserstein distance concentration bound in \cite[Section 4]{xie2021robust}, we compute the distance between the true distribution and the empirical distribution of the training samples using Monte Carlo method and use it as the radii of the uncertainty sets so that the true distributions lie in the uncertainty sets. We then use the true distributions to evaluate the performance of the proposed tests. We plot the log of the error probability as a function of testing sample size $n$. It can be seen from Fig.~\ref{fig:wass} that the kernel smoothing robust test has the best performance. The performance of the direct robust test and the Wasserstein robust test are close.

\begin{figure}[htbp]
	\centering
	\includegraphics[width=3.2in]{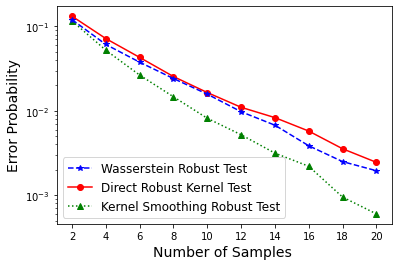}
	\caption{Comparison of the Kernel Smoothing Robust Test, the Direct Robust Kernel Test and the Wasserstein Robust Test: Synthetic Data.}
	\label{fig:wass}
\end{figure}

We then validate the performance of our robust tests using real data of human activity detection. The dataset was released by the Wireless Sensor Data Mining (WISDM) Lab in October 2013, which was collected with the Actitracker system \cite{lockhart2011, kwapisz2011activity, gary2012impact}. Users carried smartphone and were asked to do different activities. For each person, the dataset records the user's name, activities and the acceleration of the user in three directions. We use the walking data and the jogging data collected from four different users to form $H_0$ and $H_1$ respectively. We use five samples from each user to construct the uncertainty sets. The radii of the uncertainty sets are chosen by cross-validation for fair comparison. We plot the log scale error probability as a function of testing sample size $n$. In Fig.~\ref{fig:comparison_realdata}, it can be seen that the performance of the kernel smoothing robust test is better than the Wasserstein robust test and the direct robust kernel test. These results demonstrate the good performance of our kernel robust framework.

\begin{figure}[htbp]
	\centering
	\includegraphics[width=3.2in]{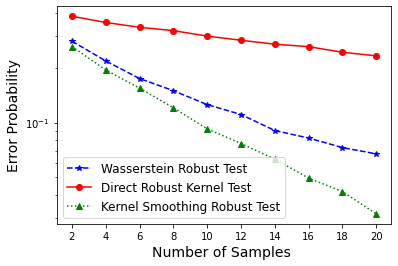}
	\caption{Comparison of the Kernel Smoothing Robust Test, the Direct Robust Kernel Test and the Wasserstein Robust Test: Human Activity Dataset.}
	\label{fig:comparison_realdata}
\end{figure}

We then compare the performance of the three algorithms using MNIST handwritten digits dataset \cite{lecun1998gradient}. We first normalize the image data and then select five images from two different classes to construct the uncertainty sets. The radii of the uncertainty sets are chosen by cross-validation for fair comparison. We plot the log scale error probability as a function of testing sample size $n$. From Fig.~\ref{fig:comparison_mnist}, it can be seen that the performance of the kernel smoothing robust test and the direct robust kernel test are close. Moreover, both the kernel smoothing robust test and the direct robust kernel test outperform the Wasserstein robust test.

\begin{figure}[htbp]
	\centering
	\includegraphics[width=3.2in]{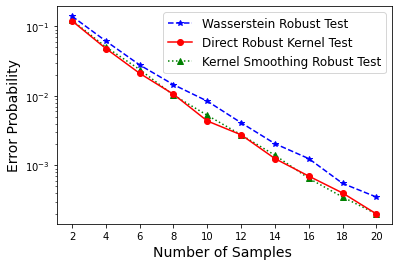}
	\caption{Comparison of the Kernel Smoothing Robust Test, the Direct Robust Kernel Test and the Wasserstein Robust Test: MNIST Handwritten Digits Dataset.}
	\label{fig:comparison_mnist}
\end{figure}

\subsection{Comparison under Different Dimensions and Training Sample Sizes}
In this section, we compare the performance of our kernel smoothing robust test under different data dimensions and different training sample sizes. 

We evaluate the performance of our kernel smoothing test when the date dimension is 5, 10, 15 and 20. When the data dimension is 5, we use 20 samples from $\mathcal{N}(\mathbf{0}, \mathbf{I})$ and $20$ samples from $\mathcal{N}(0.48\mathbf{e}, \mathbf{I})$ to construct the uncertainty sets under $H_0$ and $H_1$ respectively. When the data dimensions are 10, 15 and 20, we scale the mean of the Gaussian distribution under $H_1$ so that the KL divergence between the true distributions under $H_0$ and $H_1$ is the same for different data dimensions \cite{ramdas2015decrease}. We use a Gaussian kernel and the bandwidth parameter is chosen using cross-validation. It can be seen from Fig.~\ref{fig:dimension} that the performance of the kernel smoothing test decreases when the data dimension increases. 

\begin{figure}[htbp]
	\centering
	\includegraphics[width=3.2in]{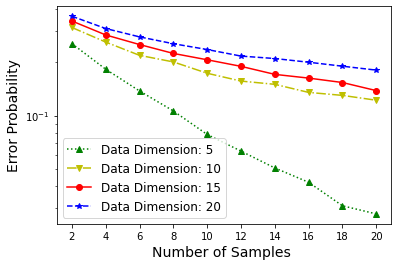}
	\caption{Comparison of the Kernel Smoothing Robust Test under Different Data Dimensions.}
	\label{fig:dimension}
\end{figure}

We then examine the impact of the training sample size on our kernel smoothing test. We use different number of training samples from $\mathcal{N}(\mathbf{0}, \mathbf{I})$ and $\mathcal{N}(0.22\mathbf{e}, \mathbf{I})$ to construct the uncertainty sets under $H_0$ and $H_1$ respectively. The data dimension is 20. The radius of the uncertainty set is chosen such that the true distributions lie in the uncertainty sets with the same probability for different training sample sizes. From Fig.~\ref{fig:train}, it can be seen that the kernel smoothing test performs better when the training sample size is larger. This validates the observation that when the training sample size is larger, we have more information about the true distributions, and the problem shall be easier to solve.

\begin{figure}[htbp]
	\centering
	\includegraphics[width=3.2in]{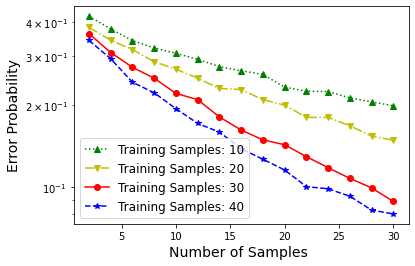}
	\caption{Comparison of the Kernel Smoothing Robust Test under Different Training Sample Sizes.}
	\label{fig:train}
\end{figure}

\subsection{Robust Kernel Test under the Neyman-Pearson Setting}\label{sec:numNP}
For the Neyman-Pearson setting, we show the good performance of our robust kernel test. We first demonstrate the performance of our tests using multivariate Gaussian distributions. For hypotheses $H_0$, we use $50$ samples generated from $\mathcal{N}(\mathbf{0},\mathbf{I})$ to construct the uncertainty set. For $H_1$, we use $50$ samples generated from $\mathcal{N}(0.5\mathbf{e}, \mathbf{I})$ to construct the uncertainty set. The data dimension is 4. The radii are chosen such that the uncertainty sets do not overlap. To test the robustness of our tests, we choose an arbitrary pair of distributions $P_0$, $P_1$ that lie on the boundary of the uncertainty sets  $\mathcal{P}_0$, $\mathcal{P}_1$. Specifically, $P_0$ and $P_1$ are multi-variate Gaussian distributions with mean $0.1\mathbf{e}$ and $0.4\mathbf{e}$, respectively, and with the same covariance matrix.

\begin{figure}[htb]
	\centering 
	\subfigure[Type-\uppercase\expandafter{\romannumeral1} error v.s. sample size]{
		\label{fig:1}
		\includegraphics[width=0.48\linewidth]{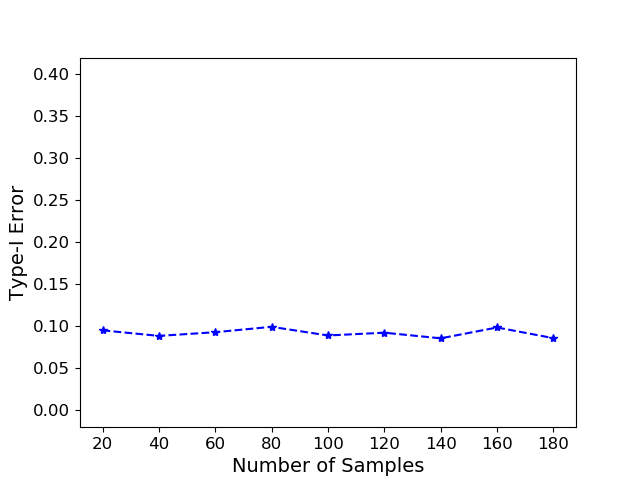}}
	\subfigure[Type-\uppercase\expandafter{\romannumeral2} error v.s. sample size.]{
		\label{fig:2}
		\includegraphics[width=0.48\linewidth]{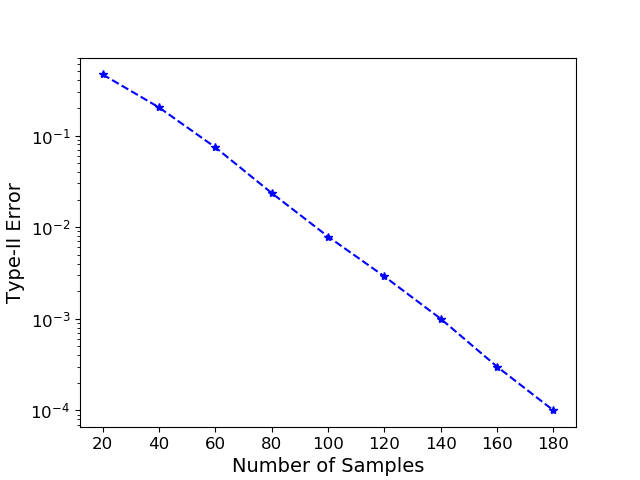}}
	\caption{Error Probability of Robust Kernel Test on Synthetic Dataset.}
	\label{Fig.1}
\end{figure}

We set the false alarm constraint $\alpha = 0.1$. With a proper choice of threshold, in Fig.~\ref{fig:1}, we plot the type-\uppercase\expandafter{\romannumeral1} error probability as a function of sample size $n$.  We repeat the experiment for 10000 times. In Fig.~\ref{fig:2}, we plot the the type-\uppercase\expandafter{\romannumeral2} error probability as a function of sample size $n$. It can be seen that the type-\uppercase\expandafter{\romannumeral2} error probability of our robust kernel test decays exponentially fast
with the sample size $n$ while the type-\uppercase\expandafter{\romannumeral1} error probability satisfies the false alarm constraint. 

We then use the real data set as in Section \ref{sec:numBay} to demonstrate the performance of our robust kernel test. We use the walking data collected from the person indexed by 685 and the walking data collected from the person indexed by 669 to form hypotheses $H_0$ and $H_1$. A small portion of the data is used to construct the uncertainty sets. The radius $\theta$ of the uncertainty sets is chosen such that the two uncertainty sets do not overlap. 

We set the false alarm constraint $\alpha = 0.1$. With a proper choice of threshold, we plot the type-I and type-II error probability as a function of sample size $n$. From Fig.~\ref{Fig.3}, it can be seen that the type-II error probability of our robust kernel test decays exponentially fast with sample size $n$ while the type-I error probability satisfies the false alarm constraint.
\begin{figure}[htb]
	\centering 
	\subfigure[Type-\uppercase\expandafter{\romannumeral1} error v.s. sample size]{
		\label{fig:5}
		\includegraphics[width=0.48\linewidth]{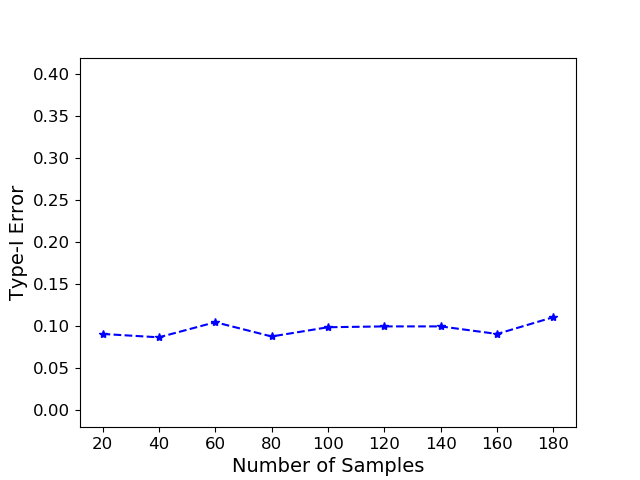}}
	\subfigure[Type-\uppercase\expandafter{\romannumeral2} error v.s. sample size.]{
		\label{fig:6}
		\includegraphics[width=0.48\linewidth]{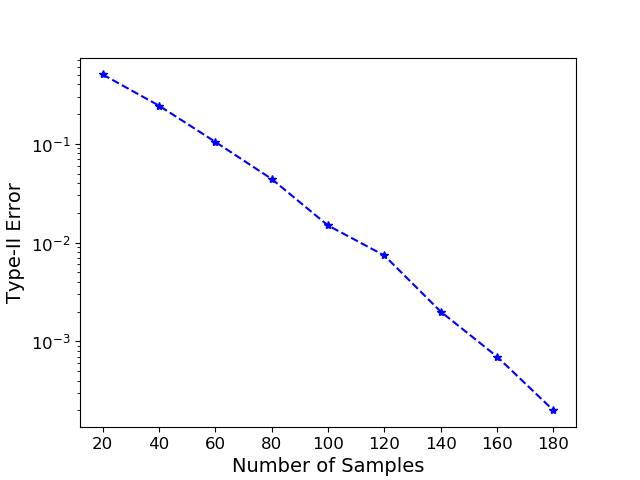}}
	\caption{Error Probability of Robust Kernel Test on the Real Dataset.}
	\label{Fig.3}
\end{figure}
	
\section{Conclusion}\label{sec:conclusion}
In this paper, we studied the robust hypothesis testing problem. We proposed a data-driven approach to construct the uncertainty sets using distance between kernel mean embeddings of distributions. 
Under the Bayesian setting, we first found the optimal test for the case with a finite alphabet. For the case with an infinite alphabet, we proposed a tractable approximation to quantify the worst-case error probability, and we developed a kernel smoothing method to generalize to unseen data in the alphabet. We also developed a direct robust kernel test which was further shown to be exponentially consistent. 
Under the Neyman-Pearson setting, we constructed a robust kernel test which can be implemented efficiently and further proved that the proposed test is asymptotically optimal. Specifically, we derived an universal upper bound on the type-\uppercase\expandafter{\romannumeral2} error exponent, and then showed that our robust kernel test achieved this universal upper bound. We also provided some numerical results to demonstrate the performance of our tests. Our approaches provide useful insights for robust hypothesis testing problems in high-dimensional setting.
	
In the future, it is of interest to investigate the robust multiple hypothesis testing problem with kernel uncertainty sets, where the design of robust detector is significantly more challenging. Another possible extension is to consider the kernel robust sequential hypothesis testing. In this case, we aim to minimize the worst-case probability of errors regarding the hypothesis using as few samples as possible, for which a data-driven approach needs to be developed. 
	
	

\appendices
\section{Useful Lemmas}\label{sec:appendix}
In this section, we list one useful lemma for our proof.

\begin{lemma}\label{lemma:KLD}\cite{hoeffding1965asymptotic, vanerven2014renyi}
For any $P_1\in\mathcal{P}$, the KL-divergence $D(\cdot\|P_1)$ is a lower semi-continuous function with respect to the weak topology of $\mathcal{P}$. That is, for any $\epsilon >0$, there exists a neighborhood $U\subset\mathcal{P}$ of $P_0$ such that for any $P^\prime\in U, D(P^\prime\|P_1)\geq D(P_0\|P_1)-\epsilon$ if $D(P_0\|P_1)<\infty$, and $D(P^\prime\|P_1)\rightarrow \infty$ as $P^\prime$ converges to $P_0$ if $D(P_0\|P_1) = \infty$.
\end{lemma}

\section{Radii Selection}
We first provide a concentration results for kernel MMD in the following lemma.
\begin{lemma}\label{lemma:empiricalmmd}\cite{altun2006unifying, szabo2015twostage}
Assume $0\leq k(\cdot,\cdot)\leq K $. Given samples $x^m = (x_1, x_2, \cdots, x_m)$ i.i.d. generated from  $P_0$, denote by $\hat{{P}}_m$ the empirical distribution of $x^m$, we then  have that 
\begin{flalign}
P_0\Big(\big|\big|\mu_{\hat{P}_m} - \mu_{P_0}\big|\big|_{\mathcal{H}} > \big(2K/m\big)^{1/2} + \epsilon\Big) \leq  \exp\Big(-\frac{\epsilon^2m}{2K}\Big).\nn
\end{flalign}
\end{lemma}
This lemma provides a method to choose the radius of the uncertainty set so that the true distribution lies in the uncertainty set with high probability. Let $0<\delta<1$, when the training sample size in $m$, to guarantee that the true distribution lies in the uncertainty set with probability at least $1-\delta$, the radius of the uncertainty set should be chosen as 
\begin{align}\label{eq:radii}
    \theta = \sqrt{\frac{2K}{m}} + \sqrt{\frac{2K\log \frac{1}{\delta}}{m}}.
\end{align}
%
This method in \eqref{eq:radii} is a straightforward approach to apply, and usually works very when there is a good number of training samples. In practice, to avoid being overly conservative, it is recommended to choose the radii using this method together with approaches, e.g., cross validation.

\section{Proof of Lemma\ref{lemma:uppersemi}}\label{sec:prooflemma1}
\begin{proof}
To prove Lemma \ref{lemma:uppersemi}, we will first show that $\int\min\big\{p_0(x), p_1(x)\big\}dx$ is concave in $p_0, p_1$. Let $B$ be the $\sigma$-field on $\mathcal{X}$. Let $\mathcal{A} = \{\mathcal{A}_1, \mathcal{A}_2,\cdots,\mathcal{A}_{|\mathcal{A}|}\}$ be a finite partition of $\mathcal{X}$ which divides $\mathcal{X}$ into a finite number of sets and $|\mathcal{A}|$ denotes the number of partitions in $\mathcal{A}$. Denote by $\Pi$ the collection of all finite B-measurable partitions. Let $P_0^{\mathcal{A}_i} = P_0(\mathcal{A}_i)$ and $P_1^{\mathcal{A}_i} = P_1(\mathcal{A}_i)$ for $i = 1, 2, \cdots, |\mathcal{A}|$.
We will then prove that $\int\min\big\{p_0(x), p_1(x)\big\}dx = \inf_{\mathcal{A}\in\Pi}\sum_{i=1}^{|\mathcal{A}|}\min\big\{P_0^{\mathcal{A}_i}, P_1^{\mathcal{A}_i}\big\}$, and show the upper semi-continuity.
	
\textbf{Step 1.} Let $f_0(p_0, p_1) =  p_0$ and $f_1(p_0, p_1) = p_1$. Since $f_0(p_0, p_1)$ and $f_1(p_0, p_1)$ are linear in $p_0, p_1$, $\min\{p_0, p_1\}$ is the minimum of two linear functions thus is concave. Therefore, $\int\min\big\{p_0(x), p_1(x)\big\}dx$ is concave in $p_0, p_1$.
	
\textbf{Step 2.} For any partitions $\mathcal{A}\in\Pi$, we have that  $\mathcal{A}_i\in B, \forall i\in \{1, 2, \cdots, |\mathcal{A}|\}$.  For any $\mathcal{A}\in\Pi$, from the concavity of $\min\big\{p_0(x),p_1(x)\big\}$ and Jensen's inequality \cite{jensen1906inequality}, we have that 
\begin{flalign}\label{eq:jensen}
\int\min\big\{p_0(x),p_1(x)\big\}dx \leq \sum_{i=1}^{|\mathcal{A}|}\min\big\{P_0^{\mathcal{A}_i}, P_1^{\mathcal{A}_i}\big\}.
\end{flalign}
	
We note that $0\leq \min\big\{\frac{p_0(x)}{p_1(x)}, 1\big\}\leq 1$. Therefore, for any $\epsilon >0$, there exists a partition $\{\mathcal{A}_1, \mathcal{A}_2,\cdots, \mathcal{A}_{|\mathcal{A}|}\}$ such that $\bigcup_{i=1}^{|\mathcal{A}|}\mathcal{A}_i = \mathcal{X}$ and $\bar{h}_i -\underline{h}_i <\epsilon$, where $\bar{h}_i = \sup_{x\in\mathcal{A}_i}\min\big\{\frac{p_0(x)}{p_1(x)}, 1\big\},\ \underline{h}_i = \inf_{x\in\mathcal{A}_i}\min\big\{\frac{p_0(x)}{p_1(x)}, 1\big\}$ for $i = 1, 2, \cdots, |\mathcal{A}|$. We then have that 
\begin{flalign}
\int_{\mathcal{A}_i}\bar{h}_i p_1(x)dx \geq \int_{\mathcal{A}_i}\min\Big\{\frac{p_0(x)}{p_1(x)}, 1\Big\}p_1(x)dx \geq \int_{\mathcal{A}_i}\underline{h}_i p_1(x)dx.
\end{flalign}
It follows that 
\begin{flalign}\label{eq:partition1}
P_1(\mathcal{A}_i)\bar{h}_i \geq \int_{\mathcal{A}_i}\min\Big\{\frac{p_0(x)}{p_1(x)}, 1\Big\}p_1(x)dx\geq P_1(\mathcal{A}_i)\underline{h}_i.
\end{flalign}
Moreover,
\begin{flalign}\label{eq:partition2}
P_1(\mathcal{A}_i)\bar{h}_i \geq P_1(\mathcal{A}_i)\min\Big\{\frac{P_0(\mathcal{A}_i)}{P_1(\mathcal{A}_i)},1\Big\}\geq P_1(\mathcal{A}_i)\underline{h}_i,
\end{flalign}
where the first inequality is because when $\bar{h}_i = 1$, $\min\big\{\frac{P_0(\mathcal{A}_i)}{P_1(\mathcal{A}_i)},1\big\}\leq 1$, and when $\bar{h}_i = \sup_{\mathcal{A}_i}\frac{p_0(x)}{p_1(x)}$, $\bar{h}_i = \frac{\bar{h}_i\int_{\mathcal{A}_i}p_1(x)dx}{\int_{\mathcal{A}_i}p_1(x)dx}\geq \frac{\int_{\mathcal{A}_i}p_0(x)dx}{\int_{\mathcal{A}_i}p_1(x)dx} = \frac{P_0(\mathcal{A}_i)}{P_1(\mathcal{A}_i)}$. The second inequality can also be proved similarly.
	
It then follows from \eqref{eq:partition1} and \eqref{eq:partition2} that 
\begin{flalign}
\bigg|P_1(\mathcal{A}_i)\min\Big\{\frac{P_0(\mathcal{A}_i)}{P_1(\mathcal{A}_i)}, 1\Big\} - \int_{\mathcal{A}_i}\min\big\{p_0(x), p_1(x)\big\}dx\bigg|\leq (\bar{h}_i-\underline{h}_i)P_1(\mathcal{A}_i)<\epsilon P_1(\mathcal{A}_i).
\end{flalign}
Therefore, 
\begin{flalign}
&\bigg|\sum_{i=1}^{|\mathcal{A}|} P_1(\mathcal{A}_i)\min\Big\{\frac{P_0(\mathcal{A}_i)}{P_1(\mathcal{A}_i)}, 1\Big\} - \int\min\big\{p_0(x), p_1(x)\big\}dx\bigg|\leq \epsilon.
\end{flalign}
We then have that 
\begin{flalign}
\sum_{i=1}^{|\mathcal{A}|} P_1(\mathcal{A}_i)\min\Big\{\frac{P_0(\mathcal{A}_i)}{P_1(\mathcal{A}_i)}, 1\Big\}\leq \int\min\big\{p_0(x), p_1(x)\big\}dx +\epsilon.
\end{flalign}
Let $\epsilon\rightarrow 0$, we have that 
\begin{flalign}\label{eq:leq}
\inf_{\mathcal{A}\in\Pi}\sum_{i=1}^{|\mathcal{A}|}\min\big\{P_0^{\mathcal{A}_i}, P_1^{\mathcal{A}_i}\big\}\leq \int\min\big\{p_0(x), p_1(x)\big\}dx.
\end{flalign}
Combining \eqref{eq:jensen} and \eqref{eq:leq}, we have that 
\begin{flalign}
\int\min\big\{p_0(x), p_1(x)\big\}dx = \inf_{\mathcal{A}\in\Pi}\sum_{i=1}^{|\mathcal{A}|}\min\big\{P_0^{\mathcal{A}_i}, P_1^{\mathcal{A}_i}\big\}.
\end{flalign}
	
\textbf{Step 3.} Let $C$ be the field of Borel sets of $\mathcal{X}$ that are sets of continuity for both $P_0$ and $P_1$. It was shown in  \cite[Theorem 1]{posner1975random} that $C$ generates $B$ in the sense that $B$ is the smallest $\sigma$-field containing $C$.
	
Let $\mathcal{A}^C = \big\{\mathcal{A}_1^C, \mathcal{A}_2^C, \cdots, \mathcal{A}^C_{|\mathcal{A}|}\big\}$ be a finite partition of $\mathcal{X}$ such that $\mathcal{A}_i^C \in C,\ \forall i\in \{1, 2, \cdots, |\mathcal{A}|\}$. Let $\Pi^C$ be the collection of all such finite partitions. Since $\Pi^C\subseteq\Pi$, we have that 
\begin{flalign}\label{eq:subset}
\inf_{\mathcal{A}^C\in\Pi^C}\sum_{i=1}^{|\mathcal{A}^C|}\min\big\{P_0^{\mathcal{A}_i^C}, P_1^{\mathcal{A}_i^C}\big\}\geq \inf_{\mathcal{A}\in\Pi}\sum_{i=1}^{|\mathcal{A}|}\min\big\{P_0^{\mathcal{A}_i}, P_1^{\mathcal{A}_i}\big\}.
\end{flalign}
	
Since $C$ generates $B$, applying Theorem D of section 13 \cite{halmos1950measure} to the measure $\nu = P_0 + P_1$, for any $\epsilon > 0$ and $\mathcal{A}\in\Pi$, we can find $E_i^\prime \in C$ such that $P_0(\mathcal{A}_i\Delta E_i^\prime) \leq \nu(\mathcal{A}_i\Delta E_i^\prime)\leq \epsilon$ and $P_1(\mathcal{A}_i\Delta E_i^\prime)\leq \nu(\mathcal{A}_i\Delta E_i^\prime)\leq \epsilon$, where $\Delta$ denotes the symmetric difference between two sets. Define 
\begin{flalign}
&E_1 = E^\prime_1, E_2 = E_2^\prime -E_1, \cdots, E_{|\mathcal{A}|-1} = E_{|\mathcal{A}|-1}^\prime-E_1\cup E_2\cup\cdots\cup E_{|\mathcal{A}|-2}, 
\end{flalign}
and
\begin{flalign}
E_{|\mathcal{A}|} = \mathcal{X}-E_1\cup E_2\cup\cdots\cup E_{|\mathcal{A}|-1}.
\end{flalign}
We have that $\nu(\mathcal{A}_1\Delta E_1) = \nu(\mathcal{A}_1\Delta E^\prime_1) \leq \epsilon$.
Since $\mathcal{A}_1, \mathcal{A}_2,\cdots, \mathcal{A}_{|\mathcal{A}|}$ are disjoint, we have that $\nu(A_2\cap E_1) \leq \epsilon$. It then follows that 
\begin{flalign}
\nu(\mathcal{A}_2\Delta E_2) = \nu\big(\mathcal{A}_2\Delta(E_2^\prime - E_1)\big)
\leq \nu(\mathcal{A}_2\Delta E_2^\prime) + \nu(\mathcal{A}_2\cap E_1)\leq 2\epsilon.
\end{flalign}
Similarly, we can show that for any $1\leq i\leq |\mathcal{A}|-1$, 
\begin{flalign}
\nu(\mathcal{A}_{i}\Delta E_{i}) = \nu\big(\mathcal{A}_{i}\Delta (E_{i}^\prime-E_1\cup E_2\cup\cdots\cup E_{i-1})\big)\leq i\epsilon
\end{flalign}
and 
\begin{flalign}
\nu(\mathcal{A}_{|\mathcal{A}|}\Delta E_{|\mathcal{A}|}) =& \nu\big((\mathcal{X} - \mathcal{A}_1\cup\mathcal{A}_2\cup\mathcal{A}_{|\mathcal{A}|-1})\Delta( \mathcal{X}-E_1\cup E_2\cup\cdots\cup E_{|\mathcal{A}|-1})\big)\nn\\ =& \nu\big((\mathcal{A}_1\cup\mathcal{A}_2\cup\mathcal{A}_{|\mathcal{A}|-1})\Delta(E_1\cup E_2\cup\cdots\cup E_{|\mathcal{A}|-1})\big)\leq \frac{|\mathcal{A}|(|\mathcal{A}|-1)}{2}\epsilon.
\end{flalign}
	
Therefore, for any $\epsilon>0$ and $\mathcal{A}\in\Pi$, there exists $E_1, E_2, \cdots, E_{|\mathcal{A}|} \in C$ such that 
\begin{flalign}
P_0(E_i)\leq P_0(\mathcal{A}_i) + \frac{|\mathcal{A}|(|\mathcal{A}|-1)}{2}\epsilon, P_1(E_i)\leq P_1(\mathcal{A}_i) + \frac{|\mathcal{A}|(|\mathcal{A}|-1)}{2}\epsilon, \forall 1\leq i\leq |\mathcal{A}|.
\end{flalign}
It then follows that there exists $\mathcal{A}^C\in\Pi^C$ such that 
\begin{flalign}
\sum_{i=1}^{|\mathcal{A}^C|}\min\big\{P_0^{\mathcal{A}_i^C}, P_1^{\mathcal{A}_i^C}\big\}\leq\inf_{\mathcal{A}\in\Pi}\sum_{i=1}^{|\mathcal{A}|}\min\big\{P_0^{\mathcal{A}_i}, P_1^{\mathcal{A}_i}\big\} + \frac{|\mathcal{A}|^2(|\mathcal{A}|-1)}{2}\epsilon.
\end{flalign}
Let $\epsilon \rightarrow 0$ and use \eqref{eq:subset}, we then have that
\begin{flalign}
\inf_{\mathcal{A}^C\in\Pi^C}\sum_{i=1}^{|\mathcal{A}^C|}\min\big\{P_0^{\mathcal{A}_i^C}, P_1^{\mathcal{A}_i^C}\big\}=\inf_{\mathcal{A}\in\Pi}\sum_{i=1}^{|\mathcal{A}|}\min\big\{P_0^{\mathcal{A}_i}, P_1^{\mathcal{A}_i}\big\}.
\end{flalign}
Therefore,
\begin{flalign}\label{eq:setcontinuous}
\int\min\big\{p_0(x), p_1(x)\big\}dx = \inf_{\mathcal{A}^C\in\Pi^C}\sum_{i=1}^{|\mathcal{A}^C|}\min\big\{P_0^{\mathcal{A}_i^C}, P_1^{\mathcal{A}_i^C}\big\}.
\end{flalign}
	
Let $P_0^s$ and $P_1^s$ be the sequence of probability distribution such that $P_0^s$ converges weakly to $P_0$ and $P_1^s$ converges weakly to $P_1$ as $s\rightarrow \infty$. Let $\mathcal{A}^C = \big\{\mathcal{A}_1^C, \mathcal{A}_2^C,\cdots\mathcal{A}_{|\mathcal{A}|}\big\}\in\Pi^C$. We then have that $\mathcal{A}^C_i \in C$ for all $i\in\{1, 2, \cdots, |\mathcal{A}|\}$. From the Portmanteau theorem \cite{achim2014probability} and the fact that $\mathcal{A}_i^C$ is a continuity set of $P_0$ and $P_1$, the weak convergence implies $P_0^s(\mathcal{A}_i^C) \rightarrow P_0(\mathcal{A}_i^C)$ and $P_1^s(\mathcal{A}_i^C) \rightarrow P_1(\mathcal{A}_i^C)$. It then follows that for any $\epsilon>0$, there exists $s_0$ such that for any $s>s_0$, 
\begin{flalign}
P_0(\mathcal{A}_i^C) + \frac{\epsilon}{|\mathcal{A}^C|} > P_0^s(\mathcal{A}_i^C), P_1(\mathcal{A}_i^C) + \frac{\epsilon}{|\mathcal{A}^C|} > P_1^s(\mathcal{A}_i^C).
\end{flalign} 
Therefore, for any $\epsilon >0$, there exists $\mathcal{A}^C \in\Pi^C$ such that
\begin{flalign}
&\int\min\big\{p_0(x), p_1(x)\big\}dx + 2\epsilon \nn\\&\geq \sum_{i=1}^{|\mathcal{A}^C|} \min\big\{P_0^{\mathcal{A}^C_i}, P_1^{\mathcal{A}^C_i}\big\} + \epsilon\nn\\&= \sum_{i=1}^{|\mathcal{A}^C|} \min\bigg\{P_0^{\mathcal{A}^C_i} + \frac{\epsilon}{|\mathcal{A}^C|}, P_1^{\mathcal{A}^C_i} + \frac{\epsilon}{|\mathcal{A}^C|}\bigg\} \nn\\&\geq\lim_{s\rightarrow \infty}\sum_{i=1}^{|\mathcal{A}^C|}\min\big\{P_0^s(\mathcal{A}^C_i), P_1^s(\mathcal{A}^C_i)\big\}\nn\\&\geq\lim_{s\rightarrow \infty}\int\min\big\{p_0^s(x), p_1^s(x)\big\}dx,
\end{flalign}
where the first inequality is from \eqref{eq:setcontinuous}, the last inequality is from \eqref{eq:jensen}. Let $\epsilon\rightarrow 0$, we have that $\int\min\big\{p_0(x), p_1(x)\big\}dx$ is upper semi-continuous in $P_0, P_1$ with respect to the weak convergence.
This completes the proof.
\end{proof}

\section{Proof of Theorem \ref{theorem:tractable}}\label{sec:pfthm2}
\begin{proof}
Since $\mathcal{X}$ is compact, $\mathcal{P}$ is tight, thus is sequentially compact with respect to the topology of weak convergence from the Prokhorov's theorem \cite{prokhorov1956}. Therefore, $\mathcal{P}$ is compact with respect to weak convergence. Therefore, $\mathcal{P}_0, \mathcal{P}_1$ are compact since $\mathcal{P}_0, \mathcal{P}_1$ are closed subsets of a compact set. We then have that the solutions to $\sup_{P_0\in\mathcal{P}_0, P_1\in\mathcal{P}_1}\int\min\big\{p_0(x), p_1(x)\big\}dx$ exist because upper semi-continuous function attains its supremum on a compact set. Let $P_0^*, P_1^*$ denote the optimal solutions to $\sup_{P_0\in\mathcal{P}_0, P_1\in\mathcal{P}_1}\int\min\big\{p_0(x), p_1(x)\big\}dx$. 


Let $\mathcal{A}_s = \{\mathcal{A}_s^1, \mathcal{A}_s^2, \cdots, \mathcal{A}_s^s\}$ be a partition of $\mathcal{X}$. We define the diameter of each partition $\mathcal{A}_s^i$ as $dia(\mathcal{A}_s^i) = \max_{x, x^\prime \in \mathcal{A}_s^i}\|x - x^\prime\|_2$. Since $\mathcal{X}$ is compact, we can choose the partition $\mathcal{A}_s$ such that $dia(\mathcal{A}_s^i) \rightarrow 0$ as $s\rightarrow \infty$ for $1\leq i\leq s$.

For any partition $\mathcal{A}_s^j$, let $x_s^j$ be an arbitrary point in $\mathcal{A}_s^j$. Denote by $P_0^s, P_1^s$ discrete distributions with $P_0^s(x_s^j) = P_0^*(\mathcal{A}_s^j), P_1^s(x_s^j) = P_1^*(\mathcal{A}_s^j)$. Let $h: \mathcal{X}\rightarrow\mathbb{R}$ be an arbitrary bounded, continuous function. Let $a_s^j = \inf_{x\in\mathcal{A}_s^j}h(x)$ and $b_s^j = \sup_{x\in\mathcal{A}_s^j}h(x)$. Since $h$ is continuous and the diameter of $\mathcal{A}_s^j$ goes to $0$ as $s\rightarrow \infty$, we then have that $\max_{j = 1,\cdots,s}(b_s^j-a_s^j)\rightarrow 0$ as $s\rightarrow\infty$. It then follows that 
\begin{flalign}
&\Big|\int hdP_0^s - \int hdP_0^*\Big| = \Big|\sum_{j=1}^s\int_{\mathcal{A}_s^j}\big(h-h(x_s^j)\big)dP_0^*\Big|\nn\\&\leq \max_{j = 1, \cdots, s}(b_s^j-a_s^j)\rightarrow 0,\ \text{as}\ s\rightarrow\infty.
\end{flalign}
Therefore, $P_0^s$ converges weakly to $P_0^*$ as $s\rightarrow \infty$. Similarly, $P_1^s$ converges weakly to $P_1^*$ as $s\rightarrow\infty$. Moreover, from Jensen's inequality \cite{jensen1906inequality}, we have that 
\begin{flalign}
&\int\min\big\{p_0^*(x), p_1^*(x)\big\}dx \leq \sum_{j=1}^s \min\big\{P_0^*(\mathcal{A}_s^j), P_1^*(\mathcal{A}_s^j)\big\}\nn\\ &= \sum_{j=1}^s\min\big\{P_0^s(x^j_s), P_1^s(x^j_s)\big\}.
\end{flalign}

Since MMD metrizes the weak convergence \cite{simon2018kernel, sripe2016weak}, for any $\epsilon>0$, there exists an integer $s_0$ such that for any $s> s_0$, $\big\|\mu_{P_0^s} - \mu_{P_0^*}\big\|_\mathcal{H}\leq \frac{\epsilon}{2}$ and $\big\|\mu_{P_1^s} - \mu_{P_1^*}\big\|_\mathcal{H}\leq \frac{\epsilon}{2}$. Therefore, from the triangle inequality \cite{gretton2012kernel}, we have that for any $s> s_0$,
\begin{flalign}
&\big\|\mu_{P_0^s} - \mu_{\hat{Q}^0_m}\big\|_\mathcal{H}\leq \big\|\mu_{P_0^s} - \mu_{P_0^*}\big\|_\mathcal{H} + \big\|\mu_{P_0^*} - \mu_{\hat{Q}^0_m}\big\|_\mathcal{H} \leq \theta + \frac{\epsilon}{2}.
\end{flalign}
Similarly, we have $\big\|\mu_{P_1^s} - \mu_{\hat{Q}^1_m}\big\|_\mathcal{H} \leq \theta + \frac{\epsilon}{2}$.

Rewrite $P_0^s = \sum_{j=1}^s \alpha_j\delta_{x_s^j}, P_1^s = \sum_{j=1}^s \beta_j\delta_{x_s^j}$, where $\delta_{x_s^j}$ denote the Dirac measure on $x_s^j$, $\alpha_j = P^*_0(\mathcal{A}_s^j), \beta_j = P^*_1(\mathcal{A}^j_s), \forall j = 1, \cdots, s$, and $\sum_{j=1}^s\alpha_j = 1, \sum_{j=1}^s\beta_j = 1$. Let $dis(x,y)$ denote a distance metric on $\mathcal{X}$ between $x$ and $y$. Note that $\{z_i\}_{i=1}^N$ are generated from a distribution $P$ supported on $\mathcal{X}$. Therefore, for any $x_s^j$, we have that $\min_{z\in\{z_i\}_{i=1}^N}dis(z, x_s^j)\rightarrow 0$ as $N\rightarrow \infty$. Therefore, there exists a sequence $\{z_{sj}^N\}_{N=1}^\infty$ such that $z_{sj}^N \in \{z_i\}_{i = 1}^N$ and $dis(z_{sj}^N, x_s^j)\rightarrow 0$ as $N\rightarrow \infty$. Assume that $z_{sj}^N$ are distinct for all $j$. We then construct the following distributions: $P_0^{s,N} = \sum_{j=1}^s\alpha_j\delta_{z_{sj}^N}, P_1^{s,N} = \sum_{j=1}^s\beta_j\delta_{z_{sj}^N}$. 
For any arbitrary bounded, continuous function $h: \mathcal{X}\rightarrow \mathbb{R}$, since $dis(z_{sj}^N, x_s^j)\rightarrow 0$ as $N\rightarrow \infty$, we have that for a fixed $s$, 

\begin{flalign}
&\Big|\int hdP_0^{s, N} - \int hdP_0^s\Big| = \Big|\sum_{j=1}^s \alpha_j\big(h(z^N_{sj}) - h(x^j_s)\big)\Big|\nn\\& \leq \max_{j = 1, \cdots, s} \big|h(z^N_{sj}) - h(x^j_s)\big| \rightarrow 0,\ \text{as}\ N \rightarrow \infty.
\end{flalign}
Therefore, we have that $P_0^{s,N}$ converges weakly to $P_0^s$ as $N\rightarrow\infty$. Similarly, $P_1^{s,N}$ converges weakly to $P_1^s$ as $N\rightarrow\infty$. Moreover, we have that 
\begin{flalign}
&\sum_{j=1}^s\min\big\{P_0^{s,N}(z_{sj}^N), P_1^{s,N}(z_{sj}^N)\big\} = \sum_{j=1}^s\min\big\{P_0^s(x^j_s), P_1^s(x^j_s)\big\}= \sum_{j=1}^s\min\{\alpha_j, \beta_j\}.
\end{flalign}

Since $P_0^s, P_1^s$ converges weakly to $P_0^*, P_1^*$, respectively, as $s\rightarrow\infty$, and $P_0^{s,N}, P_1^{s,N}$ converges weakly to $P_0^s, P_1^s$, respectively, as $N\rightarrow\infty$, we have that for any $\epsilon>0$, there exists an integer $s_0$ such that for all $s>s_0$, $\big\|\mu_{P_0^s} - \mu_{P_0^*}\big\|_\mathcal{H}\leq\frac{\epsilon}{2}$ and $\big\|\mu_{P_1^s} - \mu_{P_1^*}\big\|_\mathcal{H}\leq\frac{\epsilon}{2}$. For a fixed $s>s_0$, there exists an integer $N(s)$ such that for any $N>N(s)$, $\big\|\mu_{P_0^{s,N}} - \mu_{P_0^s}\big\|_\mathcal{H}\leq\frac{\epsilon}{2}$ and $\big\|\mu_{P_1^{s,N}} - \mu_{P_1^s}\big\|_\mathcal{H}\leq\frac{\epsilon}{2}$. Therefore, for a fixed $s>s_0$ and any $N>N(s)$, from the triangle inequality \cite{gretton2012kernel}, we have that 
\begin{flalign}
&\big\|\mu_{P_0^{s,N}} - \mu_{\hat{Q}^0_m}\big\|_\mathcal{H}\nn\\&\leq \big\|\mu_{P_0^{s,N}} - \mu_{P_0^s}\big\|_\mathcal{H} + \big\|\mu_{P_0^s} - \mu_{P_0^*}\big\|_\mathcal{H} + \big\|\mu_{P_0^*} - \mu_{\hat{Q}^0_m}\big\|_\mathcal{H} \nn\\&\leq \theta + \epsilon.
\end{flalign}
Similarly, we have that $\big\|\mu_{P_1^{s,N}} - \mu_{\hat{Q}^1_m}\big\|_\mathcal{H}\leq \theta + \epsilon$. It then follows that for large $N$,
\begin{flalign}
&\sup_{\substack{P_0^N\in\mathcal{P}:\big\|\mu_{P_0^N} - \mu_{\hat{Q}^0_m}\big\|_\mathcal{H}\leq \theta+\epsilon \\ P_1^N\in\mathcal{P}:\big\|\mu_{P_1^N} - \mu_{\hat{Q}^1_m}\big\|_\mathcal{H}\leq \theta+\epsilon\\ P_0^N, P_1^N\ \text{are supported on}\ \{z_i\}_{i=1}^N}}\sum_{i=1}^N\min\big\{P_0^N(z_i), P_1^N(z_i)\big\}\nn\\&\geq \sum_{j=1}^s\min\big\{P_0^{s,N}(z_{sj}^N), P_1^{s,N}(z_{sj}^N)\big\}\nn\\&\geq \int\min\big\{p_0^*(x), p_1^*(x)\big\}dx,
\end{flalign}
where the second inequality is from Jensen's inequality \cite{jensen1906inequality}.
Therefore, for any $\epsilon >0$,
\begin{flalign}
&\lim_{N\rightarrow \infty} \sup_{\substack{P_0^N\in\mathcal{P}:\big\|\mu_{P_0^N} - \mu_{\hat{Q}^0_m}\big\|_\mathcal{H}\leq \theta+\epsilon \\ P_1^N\in\mathcal{P}:\big\|\mu_{P_1^N} - \mu_{\hat{Q}^1_m}\big\|_\mathcal{H}\leq \theta + \epsilon\\ P_0^N, P_1^N\ \text{are supported on}\ \{z_i\}_{i=1}^N}}\sum_{i=1}^N\min\big\{P_0^N(z_i), P_1^N(z_i)\big\}\geq\int\min\big\{p_0^*(x), p_1^*(x)\big\}dx.
\end{flalign}
Moreover, we have that for any $\epsilon >0$,
\begin{flalign}\label{eq:discon}
&\sup_{\substack{P_0\in\mathcal{P}:\big\|\mu_{P_0} - \mu_{\hat{Q}^0_m}\big\|_\mathcal{H}\leq \theta+\epsilon \\P_1\in\mathcal{P}:\big\|\mu_{P_1} - \mu_{\hat{Q}^1_m}\big\|_\mathcal{H}\leq \theta+\epsilon }}\int\min\big\{p_0(x), p_1(x)\big\}dx\nn\\&\geq\lim_{N\rightarrow \infty} \sup_{\substack{P_0^N\in\mathcal{P}:\big\|\mu_{P_0^N} - \mu_{\hat{Q}^0_m}\big\|_\mathcal{H}\leq \theta +\epsilon \\ P_1^N\in\mathcal{P}:\big\|\mu_{P_1^N} - \mu_{\hat{Q}^1_m}\big\|_\mathcal{H}\leq \theta + \epsilon\\ P_0^N, P_1^N\ \text{are supported on}\ \{z_i\}_{i=1}^N}}\sum_{i=1}^N\min\big\{P_0^N(z_i), P_1^N(z_i)\big\},
\end{flalign}
which is due to the fact that the right-hand side and the left-hand side of \eqref{eq:discon} have the same objective function and the feasible region of the right-hand side is a subset of the feasible region of the left-hand side.
It then follows that
\begin{flalign}\label{eq:limit}
&\lim_{\epsilon\rightarrow 0}\sup_{\substack{P_0\in\mathcal{P}:\big\|\mu_{P_0} - \mu_{\hat{Q}^0_m}\big\|_\mathcal{H}\leq \theta+\epsilon \\P_1\in\mathcal{P}:\big\|\mu_{P_1} - \mu_{\hat{Q}^1_m}\big\|_\mathcal{H}\leq \theta+\epsilon }}\int\min\big\{p_0(x), p_1(x)\big\}dx\nn\\&\geq\lim_{\epsilon\rightarrow 0}\lim_{N\rightarrow \infty} \sup_{\substack{P_0^N\in\mathcal{P}:\big\|\mu_{P_0^N} - \mu_{\hat{Q}^0_m}\big\|_\mathcal{H}\leq \theta +\epsilon \\ P_1^N\in\mathcal{P}:\big\|\mu_{P_1^N} - \mu_{\hat{Q}^1_m}\big\|_\mathcal{H}\leq \theta + \epsilon\\ P_0^N, P_1^N\ \text{are supported on}\ \{z_i\}_{i=1}^N}}\sum_{i=1}^N\min\big\{P_0^N(z_i), P_1^N(z_i)\big\}\nn\\&\geq\int\min\big\{p_0^*(x), p_1^*(x)\big\}dx.
\end{flalign}

We will then show that all the inequality holds with equality in \eqref{eq:limit}. Recall the definition of $g(\theta)$ in \eqref{eq:defineg}:
\begin{flalign}
g(\theta) = \sup_{\substack{P_0\in\mathcal{P}:\big\|\mu_{P_0} - \mu_{\hat{Q}^0_m}\big\|_\mathcal{H}\leq \theta \\P_1\in\mathcal{P}:\big\|\mu_{P_1} - \mu_{\hat{Q}^1_m}\big\|_\mathcal{H}\leq \theta}}\int\min\big\{p_0(x), p_1(x)\big\}dx.
\end{flalign} 
We will show that $\lim_{\epsilon\rightarrow 0}g(\theta+\epsilon) = g(\theta)$, thus
\begin{flalign}
\lim_{\epsilon\rightarrow 0}\sup_{\substack{P_0\in\mathcal{P}:\big\|\mu_{P_0} - \mu_{\hat{Q}^0_m}\big\|_\mathcal{H}\leq \theta+\epsilon \\P_1\in\mathcal{P}:\big\|\mu_{P_1} - \mu_{\hat{Q}^1_m}\big\|_\mathcal{H}\leq \theta+\epsilon }}\int\min\big\{p_0(x), p_1(x)\big\}dx = \int\min\big\{p_0^*(x), p_1^*(x)\big\}dx.
\end{flalign}
It suffices to show that $g(\theta)$ is continuous in $\theta$. To show that, we will show that $g(\theta)$ is concave in $\theta$. Let $P_{0, \theta_1}, P_{1, \theta_1}$ be the optimal solutions to $g(\theta_1)$ and $P_{0, \theta_2}, P_{1, \theta_2}$ be the optimal solutions to $g(\theta_2)$. Consider $\lambda P_{0, \theta_1} + (1-\lambda) P_{0, \theta_2}, \lambda P_{1, \theta_1} + (1-\lambda) P_{1, \theta_2}$ for $0<\lambda<1$. From the triangle inequality \cite{gretton2012kernel}, we have that 
\begin{flalign}
&\|\lambda\mu_{P_{0, \theta_1}} + (1-\lambda) \mu_{P_{0, \theta_2}} - \mu_{\hat{Q}^0_m}\|_\mathcal{H}\nn\\&\leq \lambda\|\mu_{P_{0, \theta_1}}-\mu_{\hat{Q}^0_m}\|_{\mathcal{H}}  + (1-\lambda)\|\mu_{P_{0, \theta_2}} - \mu_{\hat{Q}^0_m}\|_{\mathcal{H}}\nn\\ &\leq \lambda\theta_1 + (1-\lambda)\theta_2.
\end{flalign}
Similarly, we have that $\|\lambda\mu_{P_{1, \theta_1}} + (1-\lambda) \mu_{P_{1, \theta_2}} - \mu_{\hat{Q}^1_m}\|_\mathcal{H}\leq \lambda\theta_1 + (1-\lambda)\theta_2$. Therefore, $\lambda P_{0, \theta_1} + (1-\lambda) P_{0, \theta_2}, \lambda P_{1, \theta_1} + (1-\lambda) P_{1, \theta_2}$ are feasible solutions to $g(\lambda\theta_1 + (1-\lambda)\theta_2)$. It then follows that
\begin{flalign}\label{eq:concaveg}
g(\lambda\theta_1 + (1-\lambda)\theta_2) &\geq \int\min\big\{\lambda p_{0, \theta_1}(x) + (1-\lambda)p_{0, \theta_2}(x), \lambda p_{1, \theta_1}(x) + (1-\lambda)p_{1, \theta_2}(x)\big\}dx\nn\\&\geq \lambda \int\min\big\{p_{0,\theta_1}(x), p_{1, \theta_1}(x)\big\}dx + (1-\lambda) \int\min\big\{p_{0, \theta_2}(x), p_{1, \theta_2}(x)\big\}dx\nn\\&= \lambda g(\theta_1) + (1-\lambda)g(\theta_2).
\end{flalign}
Therefore, $g(\theta)$ is concave in $\theta$, and thus is continuous in $\theta$. From \eqref{eq:limit} and the continuity of $g(\theta)$, we have that for any $\theta\geq0$,
\begin{flalign}\label{eq:epcontinue}
&\lim_{\epsilon\rightarrow 0}\sup_{\substack{P_0\in\mathcal{P}:\big\|\mu_{P_0} - \mu_{\hat{Q}^0_m}\big\|_\mathcal{H}\leq \theta+\epsilon \\P_1\in\mathcal{P}:\big\|\mu_{P_1} - \mu_{\hat{Q}^1_m}\big\|_\mathcal{H}\leq \theta+\epsilon }}\int\min\big\{p_0(x), p_1(x)\big\}dx\nn\\&=\lim_{\epsilon\rightarrow 0}\lim_{N\rightarrow \infty} \sup_{\substack{P_0^N\in\mathcal{P}:\big\|\mu_{P_0^N} - \mu_{\hat{Q}^0_m}\big\|_\mathcal{H}\leq \theta +\epsilon \\ P_1^N\in\mathcal{P}:\big\|\mu_{P_1^N} - \mu_{\hat{Q}^1_m}\big\|_\mathcal{H}\leq \theta + \epsilon\\ P_0^N, P_1^N\ \text{are supported on}\ \{z_i\}_{i=1}^N}}\sum_{i=1}^N\min\big\{P_0^N(z_i), P_1^N(z_i)\big\}\nn\\&=\int\min\big\{p_0^*(x), p_1^*(x)\big\}dx.
\end{flalign}
	
We will then show that
\begin{flalign}\label{eq:rcontinue}
&\lim_{N\rightarrow \infty} \sup_{\substack{P_0^N\in\mathcal{P}:\big\|\mu_{P_0^N} - \mu_{\hat{Q}^0_m}\big\|_\mathcal{H}\leq \theta \\ P_1^N\in\mathcal{P}:\big\|\mu_{P_1^N} - \mu_{\hat{Q}^1_m}\big\|_\mathcal{H}\leq \theta\\ P_0^N, P_1^N\ \text{are supported on}\ \{z_i\}_{i=1}^N}}\sum_{i=1}^N\min\big\{P_0^N(z_i), P_1^N(z_i)\big\}\nn\\&=\lim_{\epsilon\rightarrow 0}\lim_{N\rightarrow \infty} \sup_{\substack{P_0^N\in\mathcal{P}:\big\|\mu_{P_0^N} - \mu_{\hat{Q}^0_m}\big\|_\mathcal{H}\leq \theta +\epsilon \\ P_1^N\in\mathcal{P}:\big\|\mu_{P_1^N} - \mu_{\hat{Q}^1_m}\big\|_\mathcal{H}\leq \theta + \epsilon\\ P_0^N, P_1^N\ \text{are supported on}\ \{z_i\}_{i=1}^N}}\sum_{i=1}^N\min\big\{P_0^N(z_i), P_1^N(z_i)\big\}.
\end{flalign}

Recall the definition of $g_N(\theta)$ in \eqref{eq:definegn}: $$g_N(\theta) = \sup_{\substack{P_0^N\in\mathcal{P}:\big\|\mu_{P_0^N} - \mu_{\hat{Q}^0_m}\big\|_\mathcal{H}\leq \theta \\ P_1^N\in\mathcal{P}:\big\|\mu_{P_1^N} - \mu_{\hat{Q}^1_m}\big\|_\mathcal{H}\leq \theta\\ P_0^N, P_1^N\ \text{are supported on}\ \{z_i\}_{i=1}^N}}\sum_{i=1}^N\min\big\{P_0^N(z_i), P_1^N(z_i)\big\}.$$ Let $g^*(\theta) = \lim_{N\rightarrow \infty}g_N(\theta)$. This limit exists because for any $\theta > 0$, $\{g_N(\theta)\}_{N=1}^\infty$ is a non-decreasing sequence and has upper bound $g(\theta)$.

For any $N$, denote by $P_{0, \theta_1}^N, P_{1, \theta_1}^N$ the optimal solutions to $g_N(\theta_1)$ and $P_{0, \theta_2}^N, P_{1, \theta_2}^N$ the optimal solutions to $g_N(\theta_2)$. Consider $\lambda P_{0, \theta_1}^N + (1-\lambda)P_{0, \theta_2}^N, \lambda P_{1, \theta_1}^N + (1-\lambda)P_{1, \theta_2}^N$ for $0<\lambda<1$. We have that 
\begin{flalign}
&g^*(\lambda\theta_1 + (1-\lambda)\theta_2) = \lim_{N\rightarrow\infty}g_N(\lambda\theta_1 + (1-\lambda)\theta_2)\nn\\&\geq\lim_{N\rightarrow \infty}\sum_{i=1}^N\min\big\{\lambda P_{0,\theta_1}^N(z_i)+(1-\lambda)P_{0, \theta_2}(z_i), \lambda P_{1,\theta_1}^N(z_i)+(1-\lambda)P_{1, \theta_2}(z_i)\big\}\nn\\&\geq \lim_{N\rightarrow\infty}\lambda\sum_{i=1}^N\min\big\{P_{0, \theta_1}^N(z_i), P_{1, \theta_1}^N(z_i)\big\} + (1-\lambda)\sum_{i=1}^N\min\big\{P_{0,\theta_2}^N(z_i), P_{1, \theta_2}^N(z_i)\big\} \nn\\&= \lambda\lim_{N\rightarrow \infty}\sum_{i=1}^N\min\big\{P_{0,\theta_1}^N(z_i), P_{1, \theta_1}^N(z_i)\big\} + (1-\lambda)\lim_{N\rightarrow \infty}\sum_{i=1}^N\min\big\{P_{0,\theta_2}^N(z_i), P_{1, \theta_2}^N(z_i)\big\}\nn\\&= \lambda g^*(\theta_1) + (1-\lambda)g^*(\theta_2),
\end{flalign}
where the first equality is because the limits $\lim_{N\rightarrow \infty}\sum_{i=1}^N\min\big\{P_{0,\theta_1}^N(z_i), P_{1, \theta_1}^N(z_i)\big\}$ and $\lim_{N\rightarrow \infty}\sum_{i=1}^N\min\big\{P_{0,\theta_2}^N(z_i), P_{1, \theta_2}^N(z_i)\big\}$ exist. Therefore, $g^*(\theta)$ is concave in $\theta$, and thus is continuous in $\theta$. From the continuity of $g^*(\theta)$ and \eqref{eq:epcontinue}, we have that \eqref{eq:rcontinue} holds. This completes the proof.
\end{proof}

\section{Proof of Proposition \ref{proposition: tractable}}\label{sec:pfcompute}
\begin{proof}
Note that MMD is non-negative. If $\big\|\mu_{\hat{P}_n} - \mu_{\hat{Q}^0_m}\big\|_\mathcal{H} \leq \theta$, we have that $\hat{P}_n \in \mathcal{P}_0$. Therefore, $\inf_{P\in\mathcal{P}_0}\big\|\mu_{\hat{P}_n} - \mu_{P}\big\|_\mathcal{H} = 0$. We then consider the case that  $\big\|\mu_{\hat{P}_n} - \mu_{\hat{Q}^0_m}\big\|_\mathcal{H} > \theta$. For any $P\in\mathcal{P}_0$, $\big\|\mu_{P} - \mu_{\hat{Q}^0_m}\big\|_\mathcal{H} \leq \theta$, and thus by the triangle inequality \cite{gretton2012kernel}, we have that
\begin{flalign}
\big\|\mu_{\hat{P}_n} - \mu_{P}\big\|_\mathcal{H} &\geq \big\|\mu_{\hat{P}_n} - \mu_{\hat{Q}^0_m}\big\|_\mathcal{H}  - \big\|\mu_{P} - \mu_{\hat{Q}^0_m}\big\|_\mathcal{H}\nn\\&\geq \big\|\mu_{\hat{P}_n} - \mu_{\hat{Q}^0_m}\big\|_\mathcal{H} - \theta.
\end{flalign}
It then follows that 
\begin{flalign}\label{eq:norminf}
\inf_{P\in\mathcal{P}_0} \big\|\mu_{\hat{P}_n} - \mu_{P}\big\|_\mathcal{H} \geq \big\|\mu_{\hat{P}_n} - \mu_{\hat{Q}^0_m}\big\|_\mathcal{H} - \theta.
\end{flalign}
	
The equality in \eqref{eq:norminf} can be achieved when the following condition holds for a $P\in\mathcal{P}_0$:
\begin{flalign}
\big\|\mu_{\hat{P}_n} - \mu_{P}\big\|_\mathcal{H} = \big\|\mu_{\hat{P}_n} - \mu_{\hat{Q}^0_m}\big\|_\mathcal{H} - \theta.
\end{flalign}
We then construct such a $P$. Let $\lambda = \frac{\theta}{\big\|\mu_{\hat{P}_n} - \mu_{\hat{Q}^0_m}\big\|_\mathcal{H}}$. Since $\big\|\mu_{\hat{P}_n} - \mu_{\hat{Q}^0_m}\big\|_\mathcal{H} > \theta$, we have that $0<\lambda<1$. Let $P = \lambda\hat{P}_n + (1-\lambda)\hat{Q}^0_m$ be a linear combination of two distributions $\hat{P}_n$ and $\hat{Q}^0_m$. We then have that $P\in\mathcal{P}$ and
\begin{flalign}
\mu_P =& \int k(x,\cdot)d\big(\lambda\hat{P}_n + (1-\lambda)\hat{Q}^0_m\big)\nn\\ = &\lambda\int k(x,\cdot)d\hat{P}_n + (1-\lambda)\int k(x,\cdot)d\hat{Q}^0_m\nn\\ =& \lambda\mu_{\hat{P}_n} + (1-\lambda)\mu_{\hat{Q}^0_m}.
\end{flalign}
It then follows that 
\begin{flalign}
\big\|\mu_{\hat{P}_n} - \mu_{P}\big\|_\mathcal{H}  =& \big\|(1-\lambda)\mu_{\hat{P}_n} - (1-\lambda)\mu_{\hat{Q}^0_m}\big\|_\mathcal{H} \nn\\ = & (1-\lambda)\big\|\mu_{\hat{P}_n} - \mu_{\hat{Q}^0_m}\big\|_\mathcal{H} \nn\\=& \big\|\mu_{\hat{P}_n} - \mu_{\hat{Q}^0_m}\big\|_\mathcal{H} -\theta.
\end{flalign}
and
\begin{flalign}
\big\|\mu_{P} - \mu_{\hat{Q}^0_m}\big\|_\mathcal{H} = & \big\|\lambda\mu_{\hat{P}_n} - \lambda\mu_{\hat{Q}^0_m}\big\|_\mathcal{H} \nn\\ = & \lambda\big\|\mu_{\hat{P}_n} - \mu_{\hat{Q}^0_m}\big\|_\mathcal{H} \nn\\=& \theta.
\end{flalign}
Therefore, $P = \lambda\hat{P}_n + (1-\lambda)\hat{Q}^0_m \in \mathcal{P}_0$ and achieves the equality in \eqref{eq:norminf}.
Therefore, when $\big\|\mu_{\hat{P}_n} - \mu_{\hat{Q}^0_m}\big\|_\mathcal{H} > \theta$, we have that 
\begin{align}
\inf_{P\in\mathcal{P}_0}\big\|\mu_{\hat{P}_n} - \mu_{P}\big\|_\mathcal{H} = \big\|\mu_{\hat{P}_n} - \mu_{\hat{Q}^0_m}\big\|_\mathcal{H} - \theta,
\end{align}
From \eqref{eq:mmddef}, it follows that 
\begin{flalign}
&\inf_{P\in\mathcal{P}_0}\big\|\mu_{\hat{P}_n} - \mu_{P}\big\|_\mathcal{H} = \big\|\mu_{\hat{P}_n} - \mu_{\hat{Q}^0_m}\big\|_\mathcal{H} - \theta\nn\\&= \Big(\frac{1}{n^2}\sum_{i=1}^n\sum_{j=1}^nk(x_i, x_j) + \frac{1}{m^2}\sum_{i=1}^m\sum_{j=1}^mk(\hat{x}_{0, i}, \hat{x}_{0, j}) - \frac{2}{nm}\sum_{i=1}^n\sum_{j=1}^mk(x_i, \hat{x}_{0, j})\Big)^{1/2} - \theta. 
\end{flalign}
Following the same idea as in solving $\inf_{P\in\mathcal{P}_0}\big\|\mu_{\hat{P}_n} - \mu_{P}\big\|_\mathcal{H}$, the closed-form solution can also be derived for $\inf_{P\in\mathcal{P}_1}\big\|\mu_{\hat{P}_n} - \mu_{P}\big\|_\mathcal{H}$.
\end{proof}

\section{Proof of Theorem \ref{theorem:bayesopt}}\label{sec:pfdirect}
\begin{proof}
1) For the type-\uppercase\expandafter{\romannumeral1} error exponent, from the Sanov's theorem \cite{cover2006information}, we have that for any $P_0\in\mathcal{P}_0$,
\begin{flalign}
\lim_{n\rightarrow\infty}-\frac{1}{n}E_{P_0}[\phi_B(x^n)] &= \lim_{n\rightarrow \infty}-\frac{1}{n}\log P_0\Big(\inf_{P\in\mathcal{P}_0}\big\|\mu_{\hat{P}_n} - \mu_{P}\big\|_\mathcal{H} - \inf_{P\in\mathcal{P}_1}\big\|\mu_{\hat{P}_n} - \mu_{P}\big\|_\mathcal{H}\geq\gamma\Big)\nn\\&\geq \inf_{P^\prime\in \Gamma_0}D(P^\prime\|P_0),
\end{flalign}
where $\Gamma_0 = \Big\{P^\prime: \inf_{P\in\mathcal{P}_0}\big\|\mu_{P^\prime} - \mu_{P}\big\|_\mathcal{H} - \inf_{P\in\mathcal{P}_1}\big\|\mu_{P^\prime} - \mu_{P}\big\|_\mathcal{H} \geq \gamma\Big\}$ and $D(P^\prime\|P_0)$ denotes the KL-divergence between two distributions $P^\prime$ and $ P_0$. For any $P_0\in\mathcal{P}_0$ and $\gamma\in \Big(-\big\|\mu_{\hat{Q}_m^0} - \mu_{\hat{Q}_m^1}\big\|_\mathcal{H} + 2\theta, \big\|\mu_{\hat{Q}_m^0} - \mu_{\hat{Q}_m^1}\big\|_\mathcal{H} - 2\theta\Big)$, we have that 
\begin{flalign}\label{eq:twoopt1}
&\inf_{P\in\mathcal{P}_0}\big\|\mu_{P_0} - \mu_{P}\big\|_\mathcal{H} - \inf_{P\in\mathcal{P}_1}\big\|\mu_{P_0} - \mu_{P}\big\|_\mathcal{H} \nn\\&= -\inf_{P\in\mathcal{P}_1}\big\|\mu_{P_0} - \mu_{P}\big\|_\mathcal{H}\nn\\&= -\big\|\mu_{P_0}-\mu_{\hat{Q}_m^1}\big\|_\mathcal{H} + \theta\nn\\&\leq \big\|\mu_{P_0}-\mu_{\hat{Q}_m^0}\big\|_\mathcal{H} - \big\|\mu_{\hat{Q}_m^0} - \mu_{\hat{Q}_m^1}\big\|_\mathcal{H} + \theta\nn\\& \leq-\big\|\mu_{\hat{Q}_m^0} - \mu_{\hat{Q}_m^1}\big\|_\mathcal{H} + 2\theta\nn\\&<\gamma, 
\end{flalign}
where the first and second equalities are from Proposition \ref{proposition: tractable}, the first inequality is from the triangle inequality of MMD \cite{gretton2012kernel} and the second inequality is because $P_0\in\mathcal{P}_0$. 
We then have that for any $P_0\in\mathcal{P}_0$, when $\gamma\in \Big(-\big\|\mu_{\hat{Q}_m^0} - \mu_{\hat{Q}_m^1}\big\|_\mathcal{H} + 2\theta, \big\|\mu_{\hat{Q}_m^0} - \mu_{\hat{Q}_m^1}\big\|_\mathcal{H} - 2\theta\Big)$, $P_0\notin \Gamma_0$. Therefore, $\inf_{P^\prime\in \Gamma_0}D(P^\prime\|P_0)>0$ thus the type-\uppercase\expandafter{\romannumeral1} error probability of $\phi_B$ decreases exponentially fast with $n$. 
		
Similarly, for the type-\uppercase\expandafter{\romannumeral2} error exponent, we have that for any $P_1\in\mathcal{P}_1$,
\begin{flalign}
\lim_{n\rightarrow \infty}-\frac{1}{n}\log E_{P_1}[1-\phi_B(x^n)] &=  \lim_{n\rightarrow \infty}-\frac{1}{n}\log P_1\Big(\inf_{P\in\mathcal{P}_0}\big\|\mu_{\hat{P}_n} - \mu_{P}\big\|_\mathcal{H} - \inf_{P\in\mathcal{P}_1}\big\|\mu_{\hat{P}_n} - \mu_{P}\big\|_\mathcal{H}< \gamma\Big)\nn\\&\geq \inf_{P^\prime\in \Gamma_1}D(P^\prime\|P_1),
\end{flalign}
where $\Gamma_1 = \Big\{P^\prime: \inf_{P\in\mathcal{P}_0}\big\|\mu_{P^\prime} - \mu_{P}\big\|_\mathcal{H} - \inf_{P\in\mathcal{P}_1}\big\|\mu_{P^\prime} - \mu_{P}\big\|_\mathcal{H} \leq \gamma\Big\}$. For any $P_1\in\mathcal{P}_1$ and $\gamma\in \Big(-\big\|\mu_{\hat{Q}_m^0} - \mu_{\hat{Q}_m^1}\big\|_\mathcal{H} + 2\theta, \big\|\mu_{\hat{Q}_m^0} - \mu_{\hat{Q}_m^1}\big\|_\mathcal{H} - 2\theta\Big)$, we have that 
\begin{flalign}\label{eq:twoopt2}
&\inf_{P\in\mathcal{P}_0}\big\|\mu_{P_1} - \mu_{P}\big\|_\mathcal{H} - \inf_{P\in\mathcal{P}_1}\big\|\mu_{P_1} - \mu_{P}\big\|_\mathcal{H} \nn\\&= \inf_{P\in\mathcal{P}_0}\big\|\mu_{P_1} - \mu_{P}\big\|_\mathcal{H}\nn\\&= \big\|\mu_{P_1} -\mu_{\hat{Q}_m^0}\big\|_\mathcal{H}-\theta\nn\\&\geq\big\|\mu_{\hat{Q}_m^0} - \mu_{\hat{Q}_m^1}\big\|_\mathcal{H} - \big\|\mu_{P_1} -\mu_{\hat{Q}_m^1}\big\|_\mathcal{H} - \theta\nn\\&\geq \big\|\mu_{\hat{Q}_m^0} - \mu_{\hat{Q}_m^1}\big\|_\mathcal{H} - 2\theta\nn\\&>\gamma,
\end{flalign}
where the first and second equalities are from Proposition \ref{proposition: tractable}, the first inequality is from the triangle inequality of MMD \cite{gretton2012kernel} and the second inequality is because $P_1\in\mathcal{P}_1$.
Therefore, for any $P_1\in\mathcal{P}_1$, we have that $P_1\notin\Gamma_1$, and thus $\inf_{P^\prime\in \Gamma_1}D(P^\prime\|P_1)>0$ and the type-\uppercase\expandafter{\romannumeral2} error probability of $\phi_B$ decreases exponentially fast with $n$. Therefore, the direct robust kernel test $\phi_B$ is exponentially consistent.
		
2) We will then prove that with $\gamma\in \Big(-\big\|\mu_{\hat{Q}_m^0} - \mu_{\hat{Q}_m^1}\big\|_\mathcal{H} + 2\theta, \big\|\mu_{\hat{Q}_m^0} - \mu_{\hat{Q}_m^1}\big\|_\mathcal{H} - 2\theta\Big)$, $\phi_B$ and $\phi^\prime_B$ are equivalent. 	
When $\big\|\mu_{\hat{P}_n} - \mu_{\hat{Q}_m^0}\big\|_\mathcal{H}  \leq \theta$, we have that $\hat{P}_n\in\mathcal{P}_0$. From \eqref{eq:twoopt1}, it follows that
\begin{flalign}
&\inf_{P\in\mathcal{P}_0}\big\|\mu_{\hat{P}_n} - \mu_{P}\big\|_\mathcal{H} - \inf_{P\in\mathcal{P}_1}\big\|\mu_{\hat{P}_n} - \mu_{P}\big\|_\mathcal{H}
<\gamma.
\end{flalign}
Moreover, from triangle inequality \cite{gretton2012kernel}, we have that
\begin{flalign}
&\big\|\mu_{\hat{P}_n} - \mu_{\hat{Q}_m^0}\big\|_\mathcal{H} - \big\|\mu_{\hat{P}_n} - \mu_{\hat{Q}_m^1}\big\|_\mathcal{H}\nn\\&\leq \theta - \big\|\mu_{\hat{P}_n} - \mu_{\hat{Q}_m^1}\big\|_\mathcal{H}\nn\\&\leq \theta - \big\|\mu_{\hat{Q}_m^0} - \mu_{\hat{Q}_m^1}\big\|_\mathcal{H} +  \big\|\mu_{\hat{P}_n} - \mu_{\hat{Q}_m^0}\big\|_\mathcal{H} \nn\\&\leq - \big\|\mu_{\hat{Q}_m^0} - \mu_{\hat{Q}_m^1}\big\|_\mathcal{H} + 2\theta\nn\\&<\gamma.
\end{flalign}
Therefore, when $\big\|\mu_{\hat{P}_n} - \mu_{\hat{Q}_m^0}\big\|_\mathcal{H}  \leq \theta$, $\phi_B = \phi^\prime_B = 0$.
		
When $\big\|\mu_{\hat{P}_n} - \mu_{\hat{Q}_m^1}\big\|_\mathcal{H}  \leq \theta$, we have that $\hat{P}_n\in\mathcal{P}_1$. From \eqref{eq:twoopt2}, it follows that 
\begin{flalign}
&\inf_{P\in\mathcal{P}_0}\big\|\mu_{\hat{P}_n} - \mu_{P}\big\|_\mathcal{H} - \inf_{P\in\mathcal{P}_1}\big\|\mu_{\hat{P}_n} - \mu_{P}\big\|_\mathcal{H}> \gamma.
\end{flalign}
From the triangle inequality \cite{gretton2012kernel}, we have that
\begin{flalign}
&\big\|\mu_{\hat{P}_n} - \mu_{\hat{Q}_m^0}\big\|_\mathcal{H} - \big\|\mu_{\hat{P}_n} - \mu_{\hat{Q}_m^1}\big\|_\mathcal{H}\nn\\&\geq \big\|\mu_{\hat{P}_n} - \mu_{\hat{Q}_m^0}\big\|_\mathcal{H} - \theta\nn\\&\geq \big\|\mu_{\hat{Q}_m^0} - \mu_{\hat{Q}_m^1}\big\|_\mathcal{H} - \big\|\mu_{\hat{P}_n} - \mu_{\hat{Q}_m^1}\big\|_\mathcal{H} - \theta\nn\\&\geq\big\|\mu_{\hat{Q}_m^0} - \mu_{\hat{Q}_m^1}\big\|_\mathcal{H} -2\theta \nn\\&>\gamma.
\end{flalign}
Therefore, when $\big\|\mu_{\hat{P}_n} - \mu_{\hat{Q}_m^1}\big\|_\mathcal{H}  \leq \theta$, $\phi_B = \phi^\prime_B = 1$.
		
When $\big\|\mu_{\hat{P}_n} - \mu_{\hat{Q}_m^0}\big\|_\mathcal{H}  > \theta$ and $\big\|\mu_{\hat{P}_n} - \mu_{\hat{Q}_m^1}\big\|_\mathcal{H}  > \theta$, from Proposition \ref{proposition: tractable}, we have that 
\begin{flalign}
&\inf_{P\in\mathcal{P}_0}\big\|\mu_{\hat{P}_n} - \mu_{P}\big\|_\mathcal{H} - \inf_{P\in\mathcal{P}_1}\big\|\mu_{\hat{P}_n} - \mu_{P}\big\|_\mathcal{H}=\big\|\mu_{\hat{P}_n} - \mu_{\hat{Q}_m^0}\big\|_\mathcal{H} - \big\|\mu_{\hat{P}_n} - \mu_{\hat{Q}_m^1}\big\|_\mathcal{H}.
\end{flalign}
Combining the three different cases, we have that $\phi_B$ and $\phi^\prime_B$ are equivalent. We note that $\phi^\prime_B$ only consists of MMD between empirical distributions, thus it can be implemented efficiently. From \eqref{eq:mmddef}, we have that 
\begin{flalign}
&\big\|\mu_{\hat{P}_n} - \mu_{\hat{Q}^0_m}\big\|_\mathcal{H} - \big\|\mu_{\hat{P}_n} - \mu_{\hat{Q}^1_m}\big\|_\mathcal{H}\nn\\&= \Big(\frac{1}{n^2}\sum_{i=1}^n\sum_{j=1}^nk(x_i, x_j) + \frac{1}{m^2}\sum_{i=1}^m\sum_{j=1}^mk(\hat{x}_{0, i}, \hat{x}_{0, j}) - \frac{2}{nm}\sum_{i=1}^n\sum_{j=1}^mk(x_i, \hat{x}_{0, j})\Big)^{1/2}\nn\\ &- \Big(\frac{1}{n^2}\sum_{i=1}^n\sum_{j=1}^nk(x_i, x_j) + \frac{1}{m^2}\sum_{i=1}^m\sum_{j=1}^mk(\hat{x}_{1, i}, \hat{x}_{1, j}) - \frac{2}{nm}\sum_{i=1}^n\sum_{j=1}^mk(x_i, \hat{x}_{1, j})\Big)^{1/2}. 
\end{flalign}
This completes the proof.
\end{proof}

\section{Proof of Proposition \ref{proposition: worstbound}}\label{sec:pfworstbound}

\begin{proof}
	For any $P_0\in\mathcal{P}_0$, we have that
	\begin{flalign}\label{eq:mcdiarmid}
	&E_{P_0}[\phi_B(x^n)]\nn\\ &= P_0\Big(\big\|\mu_{\hat{P}_n} - \mu_{\hat{Q}_m^0}\big\|_\mathcal{H} - \big\|\mu_{\hat{P}_n} - \mu_{\hat{Q}_m^1}\big\|_\mathcal{H}\geq 0\Big)\nn\\&=P_0\Big(\big\|\mu_{\hat{P}_n} - \mu_{\hat{Q}_m^0}\big\|_\mathcal{H}^2 - \big\|\mu_{\hat{P}_n} - \mu_{\hat{Q}_m^1}\big\|_\mathcal{H}^2\geq 0\Big) \nn\\&= P_0\Big(E_{x\sim \hat{P}_n, x^\prime\sim \hat{P}_n}[k(x,x^\prime)] + E_{y\sim \hat{Q}_m^0, y^\prime\sim\hat{Q}_m^0}[k(y,y^\prime)] - 2E_{x\sim \hat{P}_n, y\sim \hat{Q}_m^0}[k(x,y)] - E_{x\sim \hat{P}_n, x^\prime\sim \hat{P}_n}[k(x,x^\prime)]\nn\\ &\hspace{1.17cm}-E_{y\sim \hat{Q}_m^1, y^\prime\sim\hat{Q}_m^1}[k(y,y^\prime)] + 2E_{x\sim \hat{P}_n, y\sim \hat{Q}_m^1}[k(x,y)] \geq 0\Big)\nn\\&= P_0\Big(-\frac{2}{n}\sum_{j=1}^n\frac{1}{m}\sum_{i=1}^m k(x_j, \hat{x}_{0,i}) + \frac{2}{n}\sum_{j=1}^n\frac{1}{m}\sum_{i=1}^m k(x_j, \hat{x}_{1,i}) + E_{y\sim \hat{Q}_m^0, y^\prime\sim\hat{Q}_m^0}[k(y,y^\prime)] \nn\\&\hspace{1.17cm}- E_{y\sim \hat{Q}_m^1, y^\prime\sim\hat{Q}_m^1}[k(y,y^\prime)] \geq 0\Big).
	\end{flalign}
	We will then bound the error probability in \eqref{eq:mcdiarmid} using the McDiarmid's inequality \cite{mcdiarmid1989}. Define 
	\begin{flalign}
	F(x_1, x_2,\cdots, x_n) = -\frac{2}{n}\sum_{j=1}^n\frac{1}{m}\sum_{i=1}^m k(x_j, \hat{x}_{0,i}) + \frac{2}{n}\sum_{j=1}^n\frac{1}{m}\sum_{i=1}^m k(x_j, \hat{x}_{1, i}). 
	\end{flalign}
	To apply the McDiarmid's inequality \cite{mcdiarmid1989}, we first need to bound 
	\begin{flalign}
	\sup_{x_1, \cdots, x_j, \cdots, x_n, x^\prime_j \in \mathcal{X}}\Big(F(x_1, \cdots, x_j, \cdots, x_n) - F(x_1, \cdots, x^\prime_j, \cdots, x_n)\Big), \forall 1\leq j\leq n.
	\end{flalign}
	It can be shown that 
	\begin{flalign}
	&\sup_{x_1, \cdots, x_j, \cdots, x_n, x^\prime_j \in \mathcal{X}}\Big(F(x_1, \cdots, x_j, \cdots, x_n) - F(x_1, \cdots, x^\prime_j, \cdots, x_n)\Big)\nn\\ 
	&=\sup_{x_1, \cdots, x_j, \cdots, x_n, x^\prime_j \in \mathcal{X}}\frac{2}{n}\Big(\frac{1}{m}\sum_{i=1}^m\big(k(x^\prime_j, \hat{x}_{0,i}) - k(x_j, \hat{x}_{0,i})\big) + \frac{1}{m}\sum_{i=1}^m\big(k(x_j, \hat{x}_{1,i}) - k(x^\prime_j, \hat{x}_{1,i})\big)\Big)\nn\\&\leq \frac{4K}{n},
	\end{flalign}
	where the last inequality is due to the fact that the kernel $k(\cdot, \cdot)$ is bounded. We will then consider the expectation of $F(x_1, x_2,\cdots, x_n)$. From similar steps in \eqref{eq:mcdiarmid}, it follows that 
	\begin{flalign}
	&E_{P_0}\big[F(x_1, x_2,\cdots, x_n)\big] + E_{y\sim \hat{Q}_m^0, y^\prime\sim\hat{Q}_m^0}[k(y,y^\prime)] - E_{y\sim \hat{Q}_m^1, y^\prime\sim\hat{Q}_m^1}[k(y,y^\prime)] \nn\\
	&= \big\|\mu_{P_0}-\mu_{\hat{Q}_m^0}\big\|^2_\mathcal{H} - \big\|\mu_{P_0}-\mu_{\hat{Q}_m^1}\big\|^2_\mathcal{H}.
	\end{flalign}
	Note that for any $P_0\in\mathcal{P}_0$, we have that $\big\|\mu_{P_0}-\mu_{\hat{Q}_m^0}\big\|^2_\mathcal{H} - \big\|\mu_{P_0}-\mu_{\hat{Q}_m^1}\big\|^2_\mathcal{H}\leq 0$. It then follows that 
	\begin{flalign}
	&P_0\Big(-\frac{2}{n}\sum_{j=1}^n\frac{1}{m}\sum_{i=1}^m k(x_j, \hat{x}_{0,i}) + \frac{2}{n}\sum_{j=1}^n\frac{1}{m}\sum_{i=1}^m k(x_j, \hat{x}_{1,i}) + E_{y\sim \hat{Q}_m^0, y^\prime\sim\hat{Q}_m^0}[k(y,y^\prime)]\nn\\&\hspace{0.65cm} - E_{y\sim \hat{Q}_m^1, y^\prime\sim\hat{Q}_m^1}[k(y,y^\prime)] \geq 0\Big)\nn\\&= P_0 \bigg(-\frac{2}{n}\sum_{j=1}^n\frac{1}{m}\sum_{i=1}^m k(x_j, \hat{x}_{0,i}) + \frac{2}{n}\sum_{j=1}^n\frac{1}{m}\sum_{i=1}^m k(x_j, \hat{x}_{1,i})\nn\\&\hspace{1.22cm} - E_{P_0}\Big[-\frac{2}{n}\sum_{j=1}^n\frac{1}{m}\sum_{i=1}^m k(x_j, \hat{x}_{0,i}) + \frac{2}{n}\sum_{j=1}^n\frac{1}{m}\sum_{i=1}^m k(x_j, \hat{x}_{1,i})\Big]\nn\\&\hspace{1.22cm}\geq E_{y\sim \hat{Q}_m^1, y^\prime\sim\hat{Q}_m^1}[k(y,y^\prime)] - E_{y\sim \hat{Q}_m^0, y^\prime\sim\hat{Q}_m^0}[k(y,y^\prime)] \nn\\&\hspace{1.22cm} -E_{P_0}\Big[-\frac{2}{n}\sum_{j=1}^n\frac{1}{m}\sum_{i=1}^m k(x_j, \hat{x}_{0,i}) + \frac{2}{n}\sum_{j=1}^n\frac{1}{m}\sum_{i=1}^m k(x_j, \hat{x}_{1,i})\Big] \bigg) \nn\\ &= P_0 \bigg(-\frac{2}{n}\sum_{j=1}^n\frac{1}{m}\sum_{i=1}^m k(x_j, \hat{x}_{0,i}) + \frac{2}{n}\sum_{j=1}^n\frac{1}{m}\sum_{i=1}^m k(x_j, \hat{x}_{1,i})\nn\\&\hspace{1.22cm} - E_{P_0}\Big[-\frac{2}{n}\sum_{j=1}^n\frac{1}{m}\sum_{i=1}^m k(x_j, \hat{x}_{0,i}) + \frac{2}{n}\sum_{j=1}^n\frac{1}{m}\sum_{i=1}^m k(x_j, \hat{x}_{1,i})\Big]\nn\\&\hspace{1.22cm}\geq \big\|\mu_{P_0}-\mu_{\hat{Q}_m^1}\big\|^2_\mathcal{H} - \big\|\mu_{P_0}-\mu_{\hat{Q}_m^0}\big\|^2_\mathcal{H}\bigg)\nn\\&\leq \exp\Bigg(-\frac{n\Big(\big\|\mu_{P_0}-\mu_{\hat{Q}_m^1}\big\|^2_\mathcal{H} - \big\|\mu_{P_0}-\mu_{\hat{Q}_m^0}\big\|^2_\mathcal{H}\Big)^2}{8K^2}\Bigg),
	\end{flalign}
	where the last inequality is from the McDiarmid's inequality \cite{mcdiarmid1989}. We then have that
	\begin{flalign}\label{eq:worstbound}
	&\sup_{P_0\in\mathcal{P}_0}E_{P_0}[\phi_B(x^n)]\leq \sup_{P_0\in\mathcal{P}_0}\exp\Bigg(-\frac{n\Big(\big\|\mu_{P_0}-\mu_{\hat{Q}_m^1}\big\|^2_\mathcal{H} - \big\|\mu_{P_0}-\mu_{\hat{Q}_m^0}\big\|^2_\mathcal{H}\Big)^2}{8K^2}\Bigg).
	\end{flalign}
	Since $\big\|\mu_{P_0}-\mu_{\hat{Q}_m^1}\big\|^2_\mathcal{H} - \big\|\mu_{P_0}-\mu_{\hat{Q}_m^0}\big\|^2_\mathcal{H}\geq 0$ and the exponential function is monotonically increasing, the optimization problem on the right-hand side of \eqref{eq:worstbound} can be solved by solving $\inf_{P_0\in\mathcal{P}_0}\big\|\mu_{P_0}-\mu_{\hat{Q}_m^1}\big\|^2_\mathcal{H} - \big\|\mu_{P_0}-\mu_{\hat{Q}_m^0}\big\|^2_\mathcal{H}$. We then have that 
	\begin{flalign}
	&\inf_{P_0\in\mathcal{P}_0}\big\|\mu_{P_0}-\mu_{\hat{Q}_m^1}\big\|^2_\mathcal{H} - \big\|\mu_{P_0}-\mu_{\hat{Q}_m^0}\big\|^2_\mathcal{H}\nn\\ &= \inf_{P_0\in\mathcal{P}_0} \Big(\big\|\mu_{P_0}-\mu_{\hat{Q}_m^1}\big\|_\mathcal{H} +\big\|\mu_{P_0}-\mu_{\hat{Q}_m^0}\big\|_\mathcal{H} \Big)\Big(\big\|\mu_{P_0}-\mu_{\hat{Q}_m^1}\big\|_\mathcal{H} - \big\|\mu_{P_0}-\mu_{\hat{Q}_m^0}\big\|_\mathcal{H}\Big).
	\end{flalign}
	From the triangle inequality of MMD \cite{gretton2012kernel}, we have that 
	\begin{flalign}
	\big\|\mu_{P_0}-\mu_{\hat{Q}_m^1}\big\|_\mathcal{H} +\big\|\mu_{P_0}-\mu_{\hat{Q}_m^0}\big\|_\mathcal{H} \geq \big\|\mu_{\hat{Q}_m^1}-\mu_{\hat{Q}_m^0}\big\|_\mathcal{H}. 
	\end{flalign}
	Moreover, since $\big\|\mu_{P_0}-\mu_{\hat{Q}_m^0}\big\|_\mathcal{H}\leq \theta$, it can be shown that 
	\begin{flalign}
	\big\|\mu_{P_0}-\mu_{\hat{Q}_m^1}\big\|_\mathcal{H} - \big\|\mu_{P_0}-\mu_{\hat{Q}_m^0}\big\|_\mathcal{H} \geq \big\|\mu_{\hat{Q}_m^1}-\mu_{\hat{Q}_m^0}\big\|_\mathcal{H} - 2\big\|\mu_{P_0}-\mu_{\hat{Q}_m^0}\big\|_\mathcal{H}\geq \big\|\mu_{\hat{Q}_m^1}-\mu_{\hat{Q}_m^0}\big\|_\mathcal{H} - 2\theta.
	\end{flalign}
	It then follows that 
	\begin{flalign}
	&\inf_{P_0\in\mathcal{P}_0}\big\|\mu_{P_0}-\mu_{\hat{Q}_m^1}\big\|^2_\mathcal{H} - \big\|\mu_{P_0}-\mu_{\hat{Q}_m^0}\big\|^2_\mathcal{H} \geq \big\|\mu_{\hat{Q}_m^1}-\mu_{\hat{Q}_m^0}\big\|_\mathcal{H}\Big(\big\|\mu_{\hat{Q}_m^1}-\mu_{\hat{Q}_m^0}\big\|_\mathcal{H} - 2\theta\Big).
	\end{flalign}
	Therefore, we have that 
	\begin{flalign}
	&\sup_{P_0\in\mathcal{P}_0}E_{P_0}[\phi_B(x^n)]\leq \exp\Bigg(-\frac{n\Big(\big\|\mu_{\hat{Q}_m^1} - \mu_{\hat{Q}_m^0}\big\|_\mathcal{H}^2-2\theta\big\|\mu_{\hat{Q}_m^1} - \mu_{\hat{Q}_m^0}\big\|_\mathcal{H}\Big)^2}{8K^2}\Bigg).
	\end{flalign}
	Following the same idea as in the proof of \eqref{eq:type1bound}, \eqref{eq:type2bound} can also be proved. 
	This completes the proof.
\end{proof}

\section{Proof of Theorem \ref{theorem:error}}\label{sec:pfnperror}
\begin{proof}
    For any $P_0 \in \mathcal{P}_0$, we have that 
\begin{flalign*}
&P_0 \Big(\inf_{P\in\mathcal{P}_0}\big\|\mu_{\hat{P}_n} - \mu_{P}\big\|_\mathcal{H}>\gamma_n\Big)\leq P_0 \Big(\big\|\mu_{\hat{P}_n} - \mu_{P_0}\big\|_\mathcal{H}>\gamma_n\Big).
\end{flalign*}
Set $\gamma_n = \sqrt{2K/n}\big(1+\sqrt{-\log\alpha}\big)$. From Lemma \ref{lemma:empiricalmmd} in Appendix \ref{sec:appendix}, we have that for any $P_0 \in \mathcal{P}_0$,
\begin{flalign}
P_0 \Big(\big\|\mu_{\hat{P}_n} - \mu_{P_0}\big\|_\mathcal{H}>\gamma_n\Big) \leq \alpha.
\end{flalign}
We then have that 
\begin{flalign}
\sup_{P_0\in\mathcal{P}_0}P_0 \Big(\inf_{P\in\mathcal{P}_0}\big\|\mu_{\hat{P}_n} - \mu_{P}\big\|_\mathcal{H} > \gamma_n\Big)\leq \alpha.
\end{flalign}
	
Note that $\gamma_n \rightarrow 0$ as $n\rightarrow \infty$. For any $\gamma>0$, there exists an integer $n_0$ such that $\gamma_n < \gamma$ for all $n> n_0$. We then have that for large $n$, 
\begin{align}
&\Big\{P^\prime \in \mathcal{P}: \inf_{P\in\mathcal{P}_0}\big\|\mu_{P^\prime} - \mu_{P}\big\|_\mathcal{H} \leq \gamma_n\Big\} \subseteq \Big\{P^\prime \in \mathcal{P}: \inf_{P\in\mathcal{P}_0}\big\|\mu_{P^\prime} - \mu_{P}\big\|_\mathcal{H} \leq \gamma\Big\}.
\end{align}
It then follows that 
\begin{flalign}
&\inf_{P_1\in\mathcal{P}_1}\lim_{n\rightarrow\infty}-\frac{1}{n}\log P_1 \Big(\inf_{P\in\mathcal{P}_0}\big\|\mu_{\hat{P}_n} - \mu_{P}\big\|_\mathcal{H}\leq \gamma_n\Big)\nn\\&\geq \inf_{P_1\in\mathcal{P}_1}\lim_{n\rightarrow\infty}-\frac{1}{n}\log P_1 \Big(\inf_{P\in\mathcal{P}_0}\big\|\mu_{\hat{P}_n} - \mu_{P}\big\|_\mathcal{H}\leq \gamma\Big)\nn\\&\geq \inf_{P_1\in\mathcal{P}_1}\inf_{\big\{P^\prime\in\mathcal{P}: \inf_{P\in\mathcal{P}_0}\|\mu_{P^\prime} - \mu_{P}\|_\mathcal{H}\leq \gamma\big\}} D(P^\prime\|P_1),
\end{flalign}
where the last inequality is from the Sanov's theorem \cite{cover2006information} and the fact that $\{P^\prime\in\mathcal{P}: \inf_{P\in\mathcal{P}_0}\big\|\mu_{P^\prime} - \mu_{P}\big\|_\mathcal{H}\leq \gamma\}$ is closed w.r.t. the weak topology. Let 
\begin{align}
\Gamma = \{P^\prime\in\mathcal{P}: \inf_{P\in\mathcal{P}_0}\big\|\mu_{P^\prime} - \mu_{P}\big\|_\mathcal{H}\leq \gamma\}.
\end{align}
Since MMD metrizes the weak convergence on $\mathcal{P}$ \cite{simon2018kernel, sripe2016weak}, and KL divergence is lower semi-continuous with respect to the weak topology of $\mathcal{P}$ (see Lemma \ref{lemma:KLD} in Appendix \ref{sec:appendix}), we have that for any $\epsilon>0$ and $P_1\in\mathcal{P}_1$, there exists a neighborhood $U$ of $P_0$ defined by MMD such that $D(P^\prime\|P_1) \geq D(P_0\|P_1)-\epsilon$ for any $P^\prime\in U$. 

Specifically, for any $P_0\in\mathcal{P}_0$, define the neighborhood of $P_0$ with radius $\gamma$ as $U(P_0, \gamma) = \{P\in\mathcal{P}:\|\mu_P - \mu_{P_0}\|_\mathcal{H} \leq \gamma\}$. From the lower semi-continuity of KL divergence, we have that for any $\epsilon>0$ and $P_1\in\mathcal{P}_1$, there exists $\gamma(\epsilon, P_0) > 0$ such that $D(P^\prime\|P_1) \geq D(P_0\|P_1)-\epsilon$ for any $P^\prime\in U(P_0, \gamma(\epsilon, P_0))$. Therefore, for any $P^\prime\in\bigcup_{P_0\in\mathcal{P}_0}U(P_0, \gamma(\epsilon, P_0))$, there exists $P_0\in\mathcal{P}_0$ such that $D(P^\prime\|P_1) \geq D(P_0\|P_1)-\epsilon$.
	
For a given $\epsilon>0$ and $P_1\in\mathcal{P}_1$, let $\gamma^* = \min_{P_0\in\mathcal{P}_0}\gamma(\epsilon, P_0)$. Since $\gamma(\epsilon, P_0) > 0$ holds for any $P_0\in\mathcal{P}_0$ and $\mathcal{P}_0$ is a closed set, we have that $\gamma^* > 0$. Let 
\begin{flalign}
\Gamma^* = \{P^\prime\in\mathcal{P}: \inf_{P\in\mathcal{P}_0}\big\|\mu_{P^\prime} - \mu_{P}\big\|_\mathcal{H}\leq \gamma^*\}.
\end{flalign} 
We then have that 
\begin{flalign}
\Gamma^* = \bigcup_{P_0\in\mathcal{P}_0}U(P_0, \gamma^*) \subseteq\bigcup_{P_0\in\mathcal{P}_0}U(P_0, \gamma(\epsilon, P_0)).
\end{flalign} 
We then have that for any $P^\prime\in\Gamma^*$, there exists a $P_0\in\mathcal{P}_0$ such that
\begin{flalign}
D(P^\prime\|P_1) \geq D(P_0\|P_1)-\epsilon.
\end{flalign}	
It then follows that there exists $P_0\in\mathcal{P}_0$ such that
\begin{flalign}
\inf_{P^\prime \in \Gamma^*} D(P^\prime\|P_1) \geq D(P_0\|P_1) -\epsilon.
\end{flalign}
We then have that for any $\epsilon >0$,
\begin{flalign}
\inf_{P^\prime \in \Gamma^*} D(P^\prime\|P_1) \geq \inf_{P_0\in\mathcal{P}_0}D(P_0\|P_1) -\epsilon.
\end{flalign}
Since $\epsilon$ can be arbitrarily small, it then follows that 
\begin{flalign}\label{eq:lowerexp}
&\inf_{P_1\in\mathcal{P}_1}\lim_{n\rightarrow\infty}-\frac{1}{n}\log P_1 \Big(\inf_{P\in\mathcal{P}_0}\big\|\mu_{\hat{P}_n} - \mu_{P}\big\|_\mathcal{H}\leq \gamma_n\Big)\geq \inf_{P_0\in\mathcal{P}_0, P_1\in\mathcal{P}_1}D(P_0\|P_1).
\end{flalign}
Combining \eqref{eq:lowerexp} with Proposition \ref{proposition:upper}, we have that 
\begin{flalign}
&\inf_{P_1\in\mathcal{P}_1}\lim_{n\rightarrow\infty}-\frac{1}{n}\log E_{P_1}[1-\phi_N(x^n)]\nn\\&=\inf_{P_1\in\mathcal{P}_1}\lim_{n\rightarrow\infty}-\frac{1}{n}\log P_1 \Big(\inf_{P\in\mathcal{P}_0}\big\|\mu_{\hat{P}_n} - \mu_{P}\big\|_\mathcal{H}\leq \gamma_n\Big)\nn\\&=\sup_{\phi:P_F(\phi)\leq \alpha} \inf_{P_1\in\mathcal{P}_1}\lim_{n\rightarrow\infty} -\frac{1}{n} \log E_{P_1}[1-\phi(x^n)]\nn\\&= \inf_{P_0\in\mathcal{P}_0,P_1\in\mathcal{P}_1} D(P_0\|P_1).
\end{flalign}
This completes the proof.
\end{proof}
	


\section*{Acknowledgment}
The work of Z. Sun and S. Zou was supported in part by the National Science Foundation under Grants 1948165, 2106560 and 2112693.
 
\bibliographystyle{ieeetr}
\bibliography{Robust2}
	
\end{document}